\documentclass{article}

\usepackage{geometry}
\usepackage{graphicx}
\usepackage{amsmath,amsthm,amssymb}
\usepackage{bbm}
\usepackage{url}
\usepackage{tikz}
\usetikzlibrary{graphs,graphs.standard,quotes,arrows.meta}
\usepackage{xcolor}
\usepackage{natbib}
\usepackage{hyperref}
\usepackage{mathtools}
\newcommand{\E}{\mathbb{E}}

\newcommand{\codeg}{\mathrm{codeg}}
\newcommand{\match}{\mathrm{match}}
\newcommand{\pop}{\mathrm{pop}}
\newcommand{\op}{\mathrm{op}}
\DeclareMathOperator{\Deg}{Deg} 

\newcommand{\Ind}{\mathbbm{1}}

\newcommand{\ones}{\mathbf{1}}

\newcommand{\gap}{\mathrm{gap}}

\DeclareMathOperator{\diag}{diag}

\DeclareMathOperator{\sep}{sep}

\newtheorem{theorem}{Theorem}[section]
\newtheorem{assumption}[theorem]{Assumption}
\newtheorem{corollary}[theorem]{Corollary}
\newtheorem{lemma}[theorem]{Lemma}
\newtheorem{proposition}[theorem]{Proposition}
\theoremstyle{remark}
\newtheorem{remark}{Remark}

\title{Lin–Lu–Yau Ricci Reweighting in the SBM: Uniform Curvature Concentration and Finite-Horizon Tracking}
\author{Varun Kotharkar}

\begin{document}
\maketitle

\begin{abstract}
We study curvature-driven edge reweighting for community recovery in the balanced two-block stochastic block model.
Given a graph $G$ with initial weights $W^{(0)}=A$, we iteratively update edge weights by Lin--Lu--Yau
(Ollivier-type) Ricci curvature \citep{ollivier2009,LinLuYau2011}, while all transportation costs are computed in the
\emph{unweighted} graph metric, $W^{(t+1)}_{xy} := \kappa_{W^{(t)}}(x,y)\,\Ind\{\{x,y\}\in E\}$. In a moderate-density regime, we prove \emph{uniform} concentration of edge curvatures and show that a single Ricci
reweighting step produces a two-level weighting that amplifies within-block connectivity relative to across-block
connectivity. As a consequence, spectral clustering on the reweighted graph has a strictly larger population eigengap, and we obtain corresponding nonasymptotic perturbation bounds and Davis--Kahan misclustering guarantees; moreover, we give an explicit sufficient condition under which the one-step misclustering bound improves\citep{vonluxburg2007,rohe2011,lei2015}.
We further analyze a fixed finite horizon $T$ of iterated reweighting: the random iterates track a deterministic
two-weight recursion uniformly over $t\le T$, yielding a principled finite-horizon “curvature flow” interpretation
for community detection in a canonical random graph model.
\end{abstract}

\section{Introduction}\label{sec:intro}

Ollivier--Ricci curvature and its Lin--Lu--Yau (LLY) limit provide a transportation-based notion of
local geometry on graphs~\citep{ollivier2009,LinLuYau2011}. These curvatures have been used both as descriptive network
statistics and as primitives for algorithmic procedures, including graph reweighting and curvature-driven
denoising/regularization and related discrete curvature flows~\citep{weber2017,ni2019,sia2019}.
This paper studies empirical LLY curvature in the balanced two-block stochastic block model (SBM) and analyzes a simple
\emph{Ricci reweighting} scheme: each observed edge is assigned a weight derived from its empirical LLY curvature, producing
a weighted graph whose normalized Laplacian is then used for spectral clustering~\citep{vonluxburg2007,rohe2011,lei2015}.
Throughout, all curvatures are computed using the \emph{unweighted} graph distance.

Unlike prior curvature-driven heuristics and network Ricci-flow procedures that are primarily empirical and often
evolve the underlying metric/graph weights without finite-sample guarantees \citep{ni2019,sia2019,weber2017},
we give a probabilistic, nonasymptotic analysis in a canonical SBM: (i) uniform-over-edges curvature concentration,
(ii) a deterministic two-weight mean-field proxy, and (iii) finite-horizon tracking of the iterates with explicit rates.

\paragraph{Model.}
We consider the balanced SBM on $2n$ vertices with two communities of size $n$.
Edges inside a community appear with probability $p_0=p_0(n)$ and edges across communities appear
with probability $p_1=p_1(n)$, with $0<p_1<p_0<1$. Write $\bar p:=\tfrac12(p_0+p_1)$.

\paragraph{Contributions.}
\begin{itemize}
\item \textbf{Uniform curvature concentration and a two-level deterministic proxy.}
In a moderately dense regime, we prove that empirical LLY curvature $\kappa(i,j)$ concentrates
\emph{uniformly over all edges} around two deterministic levels:
a within-community level $w_{\mathrm{in}}^{(n)}$ and a cross-community level $w_{\mathrm{out}}^{(n)}$.

\item \textbf{One-step Ricci reweighting amplifies community contrast.}
Using curvature concentration, we show that one-step reweighting increases the separation between
within- and cross-community interactions in the normalized Laplacian: at the population level this
yields a larger eigengap, and at the sample level we obtain corresponding perturbation bounds and
improved spectral-clustering guarantees~\citep{rohe2011,lei2015,yu2015}.

\item \textbf{Finite-horizon iterated reweighting follows a deterministic two-scalar recursion.}
For fixed $T$, we analyze iterated curvature reweighting and show that the empirical weighted
Laplacians track deterministic ``two-weight'' benchmarks uniformly over $t\le T$.
The benchmark weights evolve via an explicit two-dimensional mean-field map; in particular,
the induced \emph{benchmark} contrast and benchmark eigengap are monotone in $t$.
\end{itemize}

\paragraph{Standing moderately dense window.}
To avoid re-stating slightly different sparsity conditions in each section, we adopt a single
standing window (Assumption~\ref{ass:MDT}) that is strong enough to support all arguments in the
paper (uniform degree/co-degree control, matching-based transport lower bounds, and stability along
a finite number of reweighting steps). Earlier one-step statements remain valid under weaker
conditions; we record such minimal requirements when it improves clarity.

\paragraph{Organization.}
Section~\ref{sec:prelim} introduces transportation distance, Ollivier and LLY curvature, and the
standing SBM window. Section~\ref{sec:curv-sbm} proves uniform LLY curvature concentration on edges.
Section~\ref{sec:one-step} analyzes one-step Ricci reweighting and its spectral consequences.
Section~\ref{sec:iterated-ricci} extends the analysis to finite-horizon iterated reweighting.

\paragraph{Scope and non-scope.}
\textbf{What we do:}
(i) we analyze $\alpha$-lazy Ollivier curvature and its Lin--Lu--Yau limit on the balanced two-block SBM~\citep{ollivier2009,LinLuYau2011};
(ii) we study the Ricci-driven edge reweighting update $W^{(t+1)}_{xy}=\kappa_{W^{(t)}}(x,y)$ on the \emph{fixed} edge set $E(G)$;
(iii) we prove uniform concentration and finite-horizon ($t\le T$) tracking by a deterministic two-weight recursion;
(iv) we quantify the spectral consequences for normalized Laplacians and the resulting misclustering guarantees~\citep{vonluxburg2007,rohe2011,lei2015,yu2015}. This design choice isolates the effect of \emph{reweighting} from metric evolution and keeps transport local/stable
over finite horizons, enabling uniform concentration and iterate stability estimates.

\paragraph{Related work.}
Coarse Ollivier--Ricci curvature was introduced for Markov chains on metric spaces by \citet{ollivier2009},
and the Lin--Lu--Yau formulation adapts this perspective to graphs and a discrete-time limit \citep{LinLuYau2011}.
Curvature has since been used as a network descriptor and as a primitive for geometric denoising and flow-like
procedures; see, e.g., Forman--Ricci curvature and geometric flows in complex networks \citep{weber2017} and
Ricci-flow-inspired community detection heuristics \citep{ni2019,sia2019}.

On the statistical side, community recovery in SBMs has a mature theory with sharp thresholds and algorithmic
phase transitions; see the survey \citet{abbe2018} and references therein. For exact recovery thresholds in
the sparse $\log n/n$ regime, see \citet{abbe2016} and related work on optimal reconstruction and belief
propagation \citep{mossel2015,mossel2016}. Spectral methods for SBMs and their perturbation analyses are also
well understood, including normalized Laplacian guarantees \citep{rohe2011,lei2015} and standard eigenspace
perturbation tools \citep{davis1970,yu2015}.

Our contribution connects these lines by providing a finite-sample analysis of curvature-driven \emph{iterated}
edge reweighting in the balanced two-block SBM: we prove uniform curvature concentration, identify an explicit
two-weight deterministic proxy, and show finite-horizon tracking of the iterates with quantitative spectral and
misclustering consequences.

\section{Preliminaries}\label{sec:prelim}

\paragraph{Asymptotic regime and constants.}
Throughout, $n\to\infty$. All probabilistic statements are with respect to the randomness of the
graph model under consideration (typically the SBM), and unless stated otherwise, probabilities
and expectations are taken \emph{conditional on the community labels} $\sigma$.
We write ``w.h.p.'' for events whose probability tends to $1$ as $n\to\infty$.
In SBM results, edge probabilities may depend on $n$: we write $p_0=p_0(n)$ and $p_1=p_1(n)$ with
$0<p_1<p_0<1$. We write $\bar p:=\tfrac12(p_0+p_1)$.
A constant is ``absolute'' if it is independent of $n$; it may depend on fixed parameters declared once
(e.g.\ a fixed horizon $T$) and on uniform model bounds (e.g.\ the constant $\rho$ in
Assumption~\ref{ass:BC}).

\paragraph{Indexing and asymptotic notation.}
For an integer $m\ge 1$ write $[m]:=\{1,\dots,m\}$.
For nonnegative sequences $(a_n),(b_n)$:
\begin{itemize}
\item $a_n = O(b_n)$ means $\exists\,C<\infty$ such that $a_n \le C b_n$ for all large $n$.
\item $a_n = o(b_n)$ means $a_n/b_n \to 0$.
\item $a_n \lesssim b_n$ means $a_n \le C b_n$ for a constant $C$ independent of $n$.
\item $a_n \gtrsim b_n$ means $b_n \lesssim a_n$.
\item $a_n \asymp b_n$ means both $a_n \lesssim b_n$ and $b_n \lesssim a_n$.
\item $a_n \sim b_n$ means $a_n/b_n \to 1$.
\end{itemize}

\paragraph{Indicators and all-ones objects.}
We write $\Ind\{\mathcal E\}$ for the indicator of an event $\mathcal E$.
Let $\ones_m\in\mathbb R^m$ denote the all-ones vector and let $J_m:=\ones_m\ones_m^\top$ be the $m\times m$
all-ones matrix. (When the dimension is clear we drop the subscript.)

\begin{table}[t]
\centering
\small
\begin{tabular}{p{0.27\textwidth}p{0.68\textwidth}}
\hline
Symbol & Meaning \\
\hline
$G=(V,E)$ & finite simple (unweighted) graph with vertex set $V$ and edge set $E$ \\
$d(\cdot,\cdot)$ & unweighted graph distance (possibly $\infty$ if disconnected) \\
$\Gamma(x)$ & neighbor set $\{v:\{x,v\}\in E\}$ \\
$D_x$ & (unweighted) degree $|\Gamma(x)|$ \\
$\Gamma_{xy}$ & common neighborhood $\Gamma(x)\cap\Gamma(y)$ \\
$D_{xy}$ & (unweighted) co-degree $|\Gamma_{xy}|$ \\
$W$ & symmetric nonnegative weight matrix supported on $E$ \\
$d_x(W)$ & weighted degree $\sum_u W_{xu}$ \\
$\Deg(W)$ & weighted degree matrix $\diag\big((d_x(W))_{x\in V}\big)$ \\
$L(W)$, $S(W)$ & normalized Laplacian $I-\Deg(W)^{-1/2}W\Deg(W)^{-1/2}$ and normalized adjacency $S(W)=\Deg(W)^{-1/2}W\Deg(W)^{-1/2}$ \\
$m_x^\alpha$, $m_{x,W}^\alpha$ & $\alpha$-lazy neighbor measures (unweighted / weighted) \\
$W_1(\cdot,\cdot)$ & $1$--Wasserstein distance w.r.t.\ $d(\cdot,\cdot)$ \\
$\kappa_\alpha(x,y)$, $\kappa(x,y)$ & $\alpha$--Ollivier curvature and Lin--Lu--Yau (LLY) curvature \\
$\sigma$, $B_1,B_2$ & SBM labels and blocks; $|B_1|=|B_2|=n$ \\
$p_0,p_1$, $\bar p$ & within-/between-block edge probabilities; $\bar p=\tfrac12(p_0+p_1)$ \\
$\varepsilon_n$, $\eta_n$ & rates $\varepsilon_n=\sqrt{\log n/(n\bar p)}$, $\eta_n=\sqrt{\log n/(n\bar p^{\,3})}$ \\
\hline
\end{tabular}
\caption{Frequently used notation}
\label{tab:notation}
\end{table}

\paragraph{Normalized Laplacian as an operator.}
For any symmetric weight matrix $W\ge 0$ supported on $E(G)$ with $d_x(W)>0$ for all $x\in V$, define
\[
\Deg(W):=\diag\big((d_x(W))_{x\in V}\big),
\qquad
L(W):=I_{|V|}-\Deg(W)^{-1/2}W\Deg(W)^{-1/2},
\]
\[
S(W):=\Deg(W)^{-1/2}W\Deg(W)^{-1/2}.
\]
(We reserve $D_x,D_{xy}$ for \emph{unweighted} degree/co-degree, and $d_x(W),\Deg(W)$ for their weighted analogues.)

\subsection{A standing SBM window used throughout}\label{subsec:window}

\begin{assumption}[Balanced contrast + not-too-dense]\label{ass:BC}
There exists a constant $\rho\in(0,1/2)$ such that for all sufficiently large $n$,
\[
\rho\ \le\ \frac{p_1}{p_0}\ \le\ 1-\rho,
\qquad\text{and}\qquad
p_0\ \le\ 1-\rho.
\]
\end{assumption}

\begin{remark}[Why Assumption~\ref{ass:BC} is imposed]
The comparability bound implies $p_0\asymp_\rho p_1\asymp_\rho \bar p$, where $\bar p:=\tfrac12(p_0+p_1)$.
The constraint $p_0\le 1-\rho$ excludes the nearly-complete-graph corner; in particular it ensures that
for typical vertex pairs, the \emph{exclusive neighborhoods} have linear size on the $n\bar p$ scale,
which is essential for matching-based transport constructions.
\end{remark}

\begin{assumption}[Moderately dense window for a fixed horizon $T$]\label{ass:MDT}
Fix $T\in\mathbb N$. Assume Assumption~\ref{ass:BC} and
\begin{equation}\label{eq:MDT}
\frac{n\,\bar p^{\,3}}{\log n}\ \longrightarrow\ \infty\qquad\text{as }n\to\infty.
\end{equation}
Equivalently, $n\,\bar p^{\,3}\gg \log n$.
\end{assumption}

\begin{remark}[Immediate consequences of Assumption~\ref{ass:MDT}]\label{rem:MDT-conseq}
Under Assumption~\ref{ass:BC}, $p_0\asymp_\rho p_1\asymp_\rho \bar p$, hence
\[
n p_0 p_1\ \asymp_\rho\ n\bar p^{\,2},
\qquad
n p_1 \bar p\ \asymp_\rho\ n\bar p^{\,2}.
\]
Since $\bar p\le 1$, \eqref{eq:MDT} implies $n\bar p^{\,2}\gg \log n$ and $n\bar p\gg \log n$.
Thus uniform degree and co-degree concentration (over all vertices/pairs) holds w.h.p., and the
matching-based arguments used for curvature lower bounds (and their iterated analogues) can be run uniformly
for $t\le T$.
\end{remark}

\paragraph{Rates.}
Define
\[
\varepsilon_n := \sqrt{\frac{\log n}{n\bar p}},
\qquad
\eta_n := \sqrt{\frac{\log n}{n\bar p^{\,3}}}=\frac{\varepsilon_n}{\bar p}.
\]
Under Assumption~\ref{ass:MDT}, $\varepsilon_n=o(\bar p)$ and $\eta_n=o(1)$.

\subsection{Graphs, measures, and transportation distance}\label{subsec:transport}

\paragraph{Graphs and distances.}
Let $G=(V,E)$ be a finite, simple graph. Unless stated otherwise, $G$ is \emph{unweighted} and all
transport distances are computed in the unweighted graph metric $d(\cdot,\cdot)$.
If $G$ is disconnected we interpret $d(x,y)=\infty$ across components, so $W_1$ may be infinite in general.
In this paper we only apply $W_1$ (and hence curvature) to pairs with $d(x,y)=1$, i.e.\ edges.

For $x\in V$, define the neighbor set and (unweighted) degree
\[
\Gamma(x):=\{v\in V:\{x,v\}\in E\},
\qquad
D_x := |\Gamma(x)|.
\]
For distinct $x,y\in V$, define the common neighborhood and (unweighted) co-degree
\[
\Gamma_{xy}:=\Gamma(x)\cap\Gamma(y),
\qquad
D_{xy}:=|\Gamma_{xy}|.
\]

\paragraph{Probability measures and couplings.}
A probability measure on $V$ is a function $m:V\to[0,1]$ such that $\sum_{v\in V} m(v)=1$.
Given two probability measures $m_1,m_2$ on $V$, a \emph{coupling} of $(m_1,m_2)$ is a probability measure
$\pi$ on $V\times V$ with marginals $m_1$ and $m_2$:
\[
\sum_{y\in V}\pi(x,y)=m_1(x),\qquad
\sum_{x\in V}\pi(x,y)=m_2(y),\qquad (x,y\in V).
\]

\paragraph{The $1$--Wasserstein distance.}
The $1$--Wasserstein distance between $m_1$ and $m_2$ is
\[
W_1(m_1,m_2)
\ :=\
\inf_{\pi}\ \sum_{x,y\in V}\pi(x,y)\,d(x,y),
\]
where the infimum is over all couplings $\pi$ of $(m_1,m_2)$ (possibly $+\infty$ if the measures charge
different components).

\paragraph{Kantorovich--Rubinstein duality.}
A function $f:V\to\mathbb R$ is $1$--Lipschitz if $|f(x)-f(y)|\le d(x,y)$ for all $x,y\in V$.
By Kantorovich--Rubinstein duality \cite[Thm.~6.1]{Villani2009},
\begin{equation}\label{eq:KR}
W_1(m_1,m_2)
\ =\
\sup_{\|f\|_{\mathrm{Lip}}\le 1}\ \sum_{v\in V} f(v)\,\big(m_1(v)-m_2(v)\big).
\end{equation}

\paragraph{Total variation.}
For probability measures $m_1,m_2$ on $V$, define the total variation distance by
\[
\|m_1-m_2\|_{\mathrm{TV}}:=\frac12\sum_{v\in V}\big|m_1(v)-m_2(v)\big|.
\]

\subsection{Ollivier--Ricci curvature and Lin--Lu--Yau curvature}\label{subsec:curv-defs}

\paragraph{$\alpha$--lazy random walk measures.}
Fix $\alpha\in[0,1)$. For $x\in V$ with $D_x\ge 1$, define
\begin{equation}\label{eq:lazy-measure}
m_x^\alpha
\ :=\
\alpha\,\delta_x\ +\ \frac{1-\alpha}{D_x}\sum_{u\in\Gamma(x)}\delta_u.
\end{equation}
(If $D_x=0$ we set $m_x^\alpha:=\delta_x$, but we never use curvature at isolated vertices.)

\paragraph{$\alpha$--Ollivier curvature.}
For distinct $x,y\in V$ with $d(x,y)<\infty$, define
\begin{equation}\label{eq:kappa-alpha}
\kappa_\alpha(x,y)
\ :=\
1-\frac{W_1(m_x^\alpha,m_y^\alpha)}{d(x,y)}.
\end{equation}

\paragraph{Lin--Lu--Yau curvature (LLY).}
Following \citet{LinLuYau2011}, define
\begin{equation}\label{eq:LLY}
\kappa(x,y)
\ :=\
\lim_{\alpha\uparrow 1}\frac{\kappa_\alpha(x,y)}{1-\alpha},
\end{equation}
where the limit exists for finite unweighted graphs \cite{LinLuYau2011}.
If $\{x,y\}\in E$, then $d(x,y)=1$ and $\kappa_\alpha(x,y)=1-W_1(m_x^\alpha,m_y^\alpha)$, hence
\begin{equation}\label{eq:LLY-edge}
\kappa(x,y)
\ =\
\lim_{\alpha\uparrow 1}\frac{1-W_1(m_x^\alpha,m_y^\alpha)}{1-\alpha},
\qquad \{x,y\}\in E.
\end{equation}

\paragraph{Weighted graphs.}
Later we consider nonnegative symmetric edge weights $W=(W_{xy})_{x,y\in V}$ supported on $E$.
When weights are present, define the weighted degree and degree matrix
\[
d_x(W):=\sum_{u\in V} W_{xu},
\qquad
\Deg(W):=\diag\big((d_x(W))_{x\in V}\big),
\]
and (for $d_x(W)>0$) the weighted neighbor measure
\[
m_{x,W}:=\frac{1}{d_x(W)}\sum_{u\in V} W_{xu}\,\delta_u,
\qquad
m_{x,W}^\alpha:=\alpha\delta_x+(1-\alpha)m_{x,W}.
\]
We denote by $\kappa_{\alpha,W}(x,y)$ and $\kappa_W(x,y)$ the corresponding $\alpha$--Ollivier and LLY curvatures,
computed using the \emph{unweighted} graph distance $d(\cdot,\cdot)$ unless stated otherwise.

\section{Curvature concentration in the SBM}\label{sec:curv-sbm}

Throughout this section, all graph-transport notions ($W_1$, $m_x^\alpha$, $\kappa_\alpha$, $\kappa$)
are as in Section~\ref{sec:prelim}.

\subsection{Model and basic quantities}\label{subsec:sbm-model}

\paragraph{Balanced two-block SBM.}
Fix a partition $\sigma:[2n]\to\{1,2\}$ with blocks $B_a:=\sigma^{-1}(a)$ and $|B_1|=|B_2|=n$.
Let $G\sim\mathrm{SBM}(2n,p_0,p_1,\sigma)$ with $0<p_1<p_0<1$, meaning that for $x<y$,
\[
\mathbbm{1}\{\{x,y\}\in E\}\ \sim\
\begin{cases}
\mathrm{Ber}(p_0), & \sigma(x)=\sigma(y),\\
\mathrm{Ber}(p_1), & \sigma(x)\neq\sigma(y),
\end{cases}
\]
independently over unordered pairs.
Write $\bar p:=\tfrac12(p_0+p_1)$ and define the common expected degree
\begin{equation}\label{eq:d0-def}
d_0\ :=\ \E[D_x\mid\sigma]\ =\ (n-1)p_0+np_1\ \asymp_\rho\ n\bar p,
\end{equation}
where the last relation uses Assumption~\ref{ass:BC}.

\subsection{Concentration tools}\label{subsec:conc-tools}

\begin{lemma}[Bernstein--Chernoff for Bernoulli sums]\label{lem:chernoff}
Let $X_1,\ldots,X_m$ be independent with $X_i\sim \mathrm{Ber}(p_i)$ and set
$S=\sum_{i=1}^m X_i$ and $\mu=\E S=\sum_i p_i$.
Then for any $\lambda\ge 0$,
\[
\Pr(S\le \mu-\lambda)\le \exp\!\Big(-\frac{\lambda^2}{2\mu}\Big),
\qquad
\Pr(S\ge \mu+\lambda)\le \exp\!\Big(-\frac{\lambda^2}{2\mu+\frac{2}{3}\lambda}\Big).
\]
\end{lemma}

\begin{proof}
Apply Bernstein's inequality to $\sum_i (X_i-p_i)$ using $|X_i-p_i|\le 1$ and
$\sum_i \mathrm{Var}(X_i)\le \sum_i p_i=\mu$. The left tail follows by applying the right-tail bound to $\mu-S$.
\end{proof}

\begin{lemma}[Degree concentration]\label{lem:degree-concen}
Assume $n\bar p\ge c\log n$ for a sufficiently large absolute constant $c>0$.
Then there exists an absolute constant $C>0$ such that with probability at least $1-n^{-8}$,
simultaneously for all $x\in[2n]$,
\[
\big|D_x-d_0\big|\ \le\ C\sqrt{n\bar p\,\log n}.
\]
In particular, on this event, $\min_x D_x\ge \tfrac12 d_0$ and $\max_x D_x\le 2d_0$ for all large $n$.
\end{lemma}

\begin{proof}
Fix $x$. The degree $D_x$ is a sum of independent Bernoulli variables with mean $d_0\le 2n\bar p$.
Apply Lemma~\ref{lem:chernoff} with $\lambda=C\sqrt{n\bar p\log n}$ and take $C$ large enough so each tail is
$\le n^{-10}$. A union bound over $2n$ vertices gives probability at least $1-n^{-8}$.
\end{proof}

\begin{lemma}[Co-degree concentration]\label{lem:codegree-concen}
Assume $n\bar p^{\,2}\ge c\log n$ for a sufficiently large absolute constant $c>0$.
Then there exists an absolute constant $C>0$ such that with probability at least $1-n^{-8}$,
simultaneously for all unordered pairs $\{x,y\}$,
\[
\big|D_{xy}-\E[D_{xy}\mid\sigma]\big|\ \le\ C\sqrt{n\bar p^{\,2}\log n}.
\]
Moreover,
\[
\E[D_{xy}\mid\sigma]=
\begin{cases}
(n-2)p_0^2 + n p_1^2, & \sigma(x)=\sigma(y),\\
(2n-2)p_0p_1, & \sigma(x)\neq\sigma(y).
\end{cases}
\]
\end{lemma}

\begin{proof}
Fix $\{x,y\}$. For each $z\notin\{x,y\}$, the indicator $\mathbbm{1}\{z\in\Gamma_{xy}\}$ is Bernoulli and
independent over $z$.
Compute its success probability by cases to obtain the stated formula and note
$\E[D_{xy}\mid\sigma]\lesssim_\rho n\bar p^2$ under Assumption~\ref{ass:BC}.
Apply Lemma~\ref{lem:chernoff} with $\lambda=C\sqrt{n\bar p^{\,2}\log n}$ and union bound over $\binom{2n}{2}$ pairs.
\end{proof}

\begin{lemma}[Blockwise degree concentration]\label{lem:blockwise-deg}
For a block $B\in\{B_1,B_2\}$ define the blockwise degree $D_x^{(B)}:=|\Gamma(x)\cap B|$.
Conditional on $\sigma$, with probability at least $1-Cn^{-8}$, uniformly over all vertices $x$,
if $\sigma(x)=1$ then
\[
D_x^{(B_1)} = (n-1)p_0 \pm C\sqrt{n\bar p\log n},
\qquad
D_x^{(B_2)} = np_1 \pm C\sqrt{n\bar p\log n},
\]
and the analogous statement holds when $\sigma(x)=2$ (swap $B_1,B_2$).
\end{lemma}

\begin{proof}
Conditionally on $\sigma$, for each fixed vertex $x$ the blockwise degrees
$D_x^{(B_{\sigma(x)})}\sim \mathrm{Bin}(n-1,p_0)$ and
$D_x^{(B_{\sigma(x)}^c)}\sim \mathrm{Bin}(n,p_1)$.
Apply Lemma~\ref{lem:chernoff} to each with deviation $C\sqrt{n\bar p\log n}$ and take a union bound over
$2n$ vertices and the two blocks.
\end{proof}

\begin{lemma}[Blockwise co-degree concentration]\label{lem:blockwise-codeg}
Assume $n\bar p^{\,2}\ge c\log n$. For a block $B\in\{B_1,B_2\}$ define
$D^{(B)}_{xy}:=|\Gamma(x)\cap\Gamma(y)\cap B|$.
Conditional on $\sigma$, with probability at least $1-Cn^{-8}$,
uniformly over all unordered pairs $\{x,y\}$ and each block $B$,
\[
D^{(B)}_{xy}
=
\E\!\left[D^{(B)}_{xy}\mid\sigma\right]\ \pm\ C\sqrt{n\bar p^{\,2}\log n}.
\]
\end{lemma}

\begin{proof}
Each $D^{(B)}_{xy}$ is a sum of independent Bernoulli indicators of events
$\{(x,z)\in E,\ (y,z)\in E\}$ over $z\in B$, with success probability
$p_{xz}p_{yz}\in\{p_0^2,p_0p_1,p_1^2\}$. Apply Lemma~\ref{lem:chernoff} and union bound over
$O(n^2)$ pairs and the two blocks.
\end{proof}

\begin{lemma}[Ratio expansion under concentration]\label{lem:ratio-expansion}
Let $b_n\to\infty$ and write $A=\mu_A+\delta_A$, $B=\mu_B+\delta_B$ with $\mu_B\ge c b_n$ and
$|\delta_A|,|\delta_B|\le C\sqrt{b_n\log n}$ for constants $c,C>0$.
Assume also $|\mu_A|\le C_A b_n$ for some $C_A<\infty$.
Then for all sufficiently large $n$,
\[
\left|\frac{A}{B}-\frac{\mu_A}{\mu_B}\right|\ \le\ C'\sqrt{\frac{\log n}{b_n}},
\]
for a constant $C'$ depending only on $(c,C,C_A)$.
\end{lemma}

\begin{proof}
Write
\[
\frac{A}{B}-\frac{\mu_A}{\mu_B}
=\frac{\delta_A\mu_B-\mu_A\delta_B}{B\mu_B}.
\]
Since $\mu_B\ge c b_n$ and $|\delta_B|\le C\sqrt{b_n\log n}$, for large $n$ we have
$|B|\ge \mu_B-|\delta_B|\ge \frac12\mu_B$.
The claimed bound follows by direct estimation.
\end{proof}

\subsection{Deterministic curvature bounds}\label{subsec:det-bounds}

\begin{lemma}[Deterministic KR upper bound]\label{lem:curv-upper}
Let $\{x,y\}\in E$. Then
\[
\kappa(x,y)\ \le\ \frac{D_{xy}+2}{D_x},
\qquad
\kappa(x,y)\ \le\ \frac{D_{xy}+2}{D_y}.
\]
\end{lemma}

\begin{proof}
We prove $\kappa(x,y)\le (D_{xy}+2)/D_x$; the bound with $D_y$ follows by symmetry.
Fix $\alpha\in[0,1)$ and set $\beta:=1-\alpha$. By Kantorovich--Rubinstein duality \eqref{eq:KR},
for any $1$--Lipschitz $f$ we have
\[
W_1(m_x^\alpha,m_y^\alpha)\ \ge\ \sum_{v\in V} f(v)\big(m_x^\alpha(v)-m_y^\alpha(v)\big).
\]
Take $f(v):=d(v,y)$, which is $1$--Lipschitz. Since $\{x,y\}\in E$, we have $d(x,y)=1$.
Under $m_y^\alpha$, the atom at $y$ has distance $0$ and every neighbor of $y$ has distance $1$, hence
$\E_{m_y^\alpha}[d(\cdot,y)]=\beta$.
Under $m_x^\alpha$, the atom at $x$ contributes $\alpha\cdot 1$; among neighbors of $x$, the vertex $y$
contributes $0$, common neighbors contribute $1$, and exclusive neighbors contribute $2$. Therefore
\[
\E_{m_x^\alpha}[d(\cdot,y)]
=\alpha+\frac{\beta}{D_x}\Big(D_{xy}\cdot 1 + (D_x-1-D_{xy})\cdot 2\Big)
=\alpha+\frac{\beta}{D_x}(2D_x-2-D_{xy}).
\]
Thus
\[
W_1(m_x^\alpha,m_y^\alpha)\ \ge\ \alpha+\beta\,\frac{D_x-2-D_{xy}}{D_x},
\]
and hence
\[
1-W_1(m_x^\alpha,m_y^\alpha)\ \le\ \beta\,\frac{D_{xy}+2}{D_x}.
\]
Divide by $\beta$ and let $\alpha\uparrow 1$ to obtain the claim.
\end{proof}

\begin{lemma}[Coupling lower bound from an injection]\label{lem:coupling}
Let $\{x,y\}\in E$ and assume $D_x\le D_y$.
Define the exclusive neighborhoods
\[
U_x:=\Gamma(x)\setminus\big(\Gamma(y)\cup\{y\}\big),
\qquad
U_y:=\Gamma(y)\setminus\big(\Gamma(x)\cup\{x\}\big).
\]
If $\phi:U_x\to U_y$ is an injection, then
\begin{equation}\label{eq:coupling-lb}
\kappa(x,y)\ \ge\
1+\frac{1}{D_y}\ -\ \frac{1}{D_y}\sum_{u\in U_x} d\big(u,\phi(u)\big)\ -\ 3\,\frac{D_y-D_x}{D_y}.
\end{equation}
\end{lemma}

\begin{proof}
Fix $\alpha\in[1/2,1)$ and write $\beta:=1-\alpha$.
Construct a coupling $\pi$ of $m_x^\alpha$ and $m_y^\alpha$ as follows.

Let $T:=U_y\setminus \phi(U_x)$, so $|T|=|U_y|-|U_x|=D_y-D_x$.

\smallskip
\noindent\emph{(i) Mass at $x$.}
Set $\pi(x,x):=\beta/D_y$ and $\pi(x,y):=\alpha-\beta/D_y$ (nonnegative since $\alpha\ge 1/2$).

\smallskip
\noindent\emph{(ii) Diagonal transport for $y$ and common neighbors.}
Set $\pi(y,y):=\beta/D_y$ and for each $z\in\Gamma_{xy}$ set $\pi(z,z):=\beta/D_y$.

\smallskip
\noindent\emph{(iii) Route exclusive neighbors of $x$ via $\phi$.}
For each $u\in U_x$, set $\pi(u,\phi(u)):=\beta/D_y$.

\smallskip
\noindent\emph{(iv) Distribute surplus to fill $T$.} Each $s\in\Gamma(x)$ has source mass $\beta/D_x$ but has used $\beta/D_y$ in (ii)--(iii),
leaving surplus $\beta(\frac1{D_x}-\frac1{D_y})\ge 0$.
Equivalently, after (ii)--(iii) each neighbor-row has assigned $\beta/D_y$ mass and must allocate the remaining
$\beta/D_x-\beta/D_y$ mass to reach the required row sum $\beta/D_x$.
Distribute this surplus uniformly over $T$ by setting, for $t\in T$,
\[
\pi(s,t):=\beta\Big(\frac1{D_x}-\frac1{D_y}\Big)\cdot \frac{1}{|T|}.
\]

\smallskip
A direct check shows the row sums equal $m_x^\alpha$ and the column sums equal $m_y^\alpha$.
For transport cost, all diagonal assignments have zero cost; $x\to y$ has cost $1$; $u\to\phi(u)$ has cost
$d(u,\phi(u))$; and for (iv), if $s\in\Gamma(x)$ and $t\in T\subseteq U_y\subseteq\Gamma(y)$ then
$d(s,t)\le 3$ via the path $s-x-y-t$. Therefore
\[
W_1(m_x^\alpha,m_y^\alpha)
\ \le\ \Big(\alpha-\frac{\beta}{D_y}\Big)
\ +\ \frac{\beta}{D_y}\sum_{u\in U_x} d(u,\phi(u))
\ +\ 3\beta\Big(1-\frac{D_x}{D_y}\Big),
\]
which implies \eqref{eq:coupling-lb} after dividing by $\beta$ and letting $\alpha\uparrow 1$.
\end{proof}

\begin{lemma}[Left-perfect matching in a random bipartite graph]\label{lem:bip-matching}
Let $H$ be a random bipartite graph on $L\cup R$ with $|L|=m$, $|R|=\ell\ge m$, and $\ell\le \Lambda m$
for some fixed $\Lambda\ge 1$, where edges are present independently with probability $p$.
Then for any $q>0$ there exists $C=C(\Lambda,q)>0$ such that if
\[
p\ell \ge C\log m,
\]
then
\[
\Pr\{\text{$H$ contains a matching saturating $L$}\}\ge 1-m^{-q}.
\]
\end{lemma}

\begin{proof}
By Hall's theorem, $H$ fails to contain a matching saturating $L$ iff there exists a nonempty
$S\subseteq L$ such that $|N(S)|\le |S|-1$.

Fix $s\in\{1,\dots,m\}$ and fix $S\subseteq L$ with $|S|=s$. If $|N(S)|\le s-1$, then there exists
$T\subseteq R$ with $|T|=s-1$ such that $N(S)\subseteq T$, equivalently there are \emph{no} edges
between $S$ and $R\setminus T$. For fixed $(S,T)$,
\[
\Pr\{N(S)\subseteq T\}=(1-p)^{s(\ell-(s-1))}\le \exp\{-p\,s(\ell-s+1)\}.
\]
Union bounding over $(S,T)$ gives
\begin{equation}\label{eq:hall-union-fixed}
\Pr(\text{Hall fails})
\le \sum_{s=1}^m \binom{m}{s}\binom{\ell}{s-1}\exp\{-p\,s(\ell-s+1)\}.
\end{equation}

Set $a:=p\ell$ and use $\ell\le \Lambda m$. We split into two ranges.

\smallskip
\noindent\emph{Range I: $1\le s\le \ell/2$.}
Then $\ell-s+1\ge \ell/2$ so $\exp\{-p\,s(\ell-s+1)\}\le \exp\{-(a/2)s\}$.
Also, for $s\ge 1$,
\[
\binom{m}{s}\le \Big(\frac{em}{s}\Big)^s,
\qquad
\binom{\ell}{s-1}\le \Big(\frac{e\ell}{\max\{1,s-1\}}\Big)^{s-1}\le \Big(\frac{2e\ell}{s}\Big)^s,
\]
(where the last inequality is valid for all $s\ge 1$, treating $s=1$ separately if desired).
Hence each summand in \eqref{eq:hall-union-fixed} is at most
\[
\Big(\frac{2e^2\,m\ell}{s^2}\,e^{-a/2}\Big)^s
\le
\Big(\frac{2e^2\Lambda}{s^2}\,m^{2}\,e^{-a/2}\Big)^s.
\]
If $a\ge C\log m$, then $e^{-a/2}\le m^{-C/2}$, so the $s$th term is at most
\[
\Big(\frac{2e^2\Lambda}{s^2}\,m^{2-C/2}\Big)^s.
\]
Choose $C\ge 4(q+3)$ (for instance) so that $2-C/2\le -(q+1)$.
Then the $s$th term is $\le c(\Lambda)^s\,m^{-(q+1)s}$ and summing over $s\ge 1$ gives
\[
\sum_{s=1}^{\lfloor \ell/2\rfloor} \binom{m}{s}\binom{\ell}{s-1}e^{-p s(\ell-s+1)}
\le m^{-(q+1)}\sum_{s\ge 1} c(\Lambda)^s m^{-qs}
\le \tfrac12 m^{-q}
\]
for all large $m$.

\smallskip
\noindent\emph{Range II: $\ell/2 < s\le m$.}
Let $k:=\ell-s+1\in\{1,\dots,\lfloor \ell/2\rfloor\}$, so $s-1=\ell-k$. Choosing $T$ of size $s-1$
is equivalent to choosing its complement $U:=R\setminus T$ of size $k$, so
\[
\binom{\ell}{s-1}=\binom{\ell}{k}.
\]
Moreover, since $s\le m$ we have $\ell-k+1=s\le m$, hence
\[
\binom{m}{s}=\binom{m}{\ell-k+1}\le \binom{m}{k}
\le \Big(\frac{em}{k}\Big)^k.
\]
Also $\binom{\ell}{k}\le (e\ell/k)^k\le (e\Lambda m/k)^k$ and
$\exp\{-p\,s(\ell-s+1)\}=\exp\{-p(\ell-k+1)k\}\le \exp\{-(a/2)k\}$ because $\ell-k+1\ge \ell/2$.
Thus the $k$th term is at most
\[
\Big(\frac{e^2\Lambda m^2}{k^2}\,e^{-a/2}\Big)^k
\le
\Big(\frac{e^2\Lambda}{k^2}\,m^{2-C/2}\Big)^k.
\]
With the same choice $C\ge 4(q+3)$, summing over $k\ge 1$ gives another $\le \tfrac12 m^{-q}$ for large $m$.

Combining both ranges yields $\Pr(\text{Hall fails})\le m^{-q}$ for all sufficiently large $m$.
\end{proof}

\subsection{Curvature concentration and population levels}\label{subsec:curv-thm}

\paragraph{Population curvature levels.}
Define
\begin{equation}\label{eq:w-in-out-n}
w_{\rm in}^{(n)}:=\frac{(n-2)p_0^2+np_1^2}{(n-1)p_0+np_1},
\qquad
w_{\rm out}^{(n)}:=\frac{(2n-2)p_0p_1}{(n-1)p_0+np_1}.
\end{equation}

\begin{lemma}[Population levels are $\asymp \bar p$]\label{lem:winwout-scale}
Assume Assumption~\ref{ass:BC}. Then there exist constants $0<c_\rho<C_\rho<\infty$ such that for all large $n$,
\[
c_\rho\,\bar p\ \le\ w_{\rm out}^{(n)}\ \le\ w_{\rm in}^{(n)}\ \le\ C_\rho\,\bar p.
\]
\end{lemma}

\begin{proof}
Under Assumption~\ref{ass:BC}, $p_0\asymp_\rho p_1\asymp_\rho \bar p$ and $d_0\asymp_\rho n\bar p$.
Both numerators in \eqref{eq:w-in-out-n} are $\asymp_\rho n\bar p^2$, hence the ratios are $\asymp_\rho \bar p$.
Moreover, since $p_1/p_0\le 1-\rho$, the gap $p_0-p_1\gtrsim_\rho \bar p$, so for all large $n$ the leading-order
comparison $(p_0^2+p_1^2)>(2p_0p_1)$ dominates the $O(1/n)$ corrections in \eqref{eq:w-in-out-n}, yielding
$w_{\rm in}^{(n)}\ge w_{\rm out}^{(n)}$.
\end{proof}

\begin{theorem}[Uniform LLY curvature concentration on edges]\label{thm:sbm-main}
Let $G\sim \mathrm{SBM}(2n,p_0,p_1,\sigma)$ with $0<p_1<p_0<1$ and assume Assumption~\ref{ass:BC}.
Assume $n\bar p^{\,2}\ge c_0\log n$ for a sufficiently large absolute constant $c_0>0$.
Then there exists $C>0$ (depending only on $\rho$) such that, with probability at least $1-n^{-2}$,
simultaneously for all edges $\{x,y\}\in E(G)$,
\[
\kappa(x,y)=
\begin{cases}
w_{\rm in}^{(n)}\ \pm\ C\,\varepsilon_n,
& \sigma(x)=\sigma(y),\\[1.2ex]
w_{\rm out}^{(n)}\ \pm\ C\,\varepsilon_n,
& \sigma(x)\neq\sigma(y),
\end{cases}
\]
where $\varepsilon_n=\sqrt{\log n/(n\bar p)}$.
Moreover, for any fixed pair $\{x,y\}$ (not depending on $n$),
\[
\kappa(x,y)
=
\begin{cases}
w_{\rm in}^{(n)} + O_p\!\big((n\bar p)^{-1/2}\big), & \sigma(x)=\sigma(y),\\
w_{\rm out}^{(n)} + O_p\!\big((n\bar p)^{-1/2}\big), & \sigma(x)\neq\sigma(y).
\end{cases}
\]
\end{theorem}

\begin{proof}
Let $\mathcal E_{\deg}$ be the event in Lemma~\ref{lem:degree-concen} and $\mathcal E_{\codeg}$ the event in
Lemma~\ref{lem:codegree-concen}. Since $n\bar p^{\,2}\ge c_0\log n$ implies $n\bar p\ge c_0\log n$,
both events occur with probability at least $1-2n^{-8}$.

Fix an unordered pair $\{x,y\}$ and (for the lower bound) assume $\{x,y\}\in E$.
Write $D_x=d_0+\delta_x$, $D_y=d_0+\delta_y$, and $D_{xy}=\mu_{xy}+\eta_{xy}$ with
$\mu_{xy}=\E[D_{xy}\mid\sigma]$.
On $\mathcal E_{\deg}\cap\mathcal E_{\codeg}$,
\[
|\delta_x|+|\delta_y|\ \lesssim\ \sqrt{n\bar p\log n},
\qquad
|\eta_{xy}|\ \lesssim\ \sqrt{n\bar p^{\,2}\log n}.
\]

\smallskip\noindent
\emph{Upper bound.}
If $\{x,y\}\in E$, Lemma~\ref{lem:curv-upper} gives $\kappa(x,y)\le (D_{xy}+2)/D_x$.
Apply Lemma~\ref{lem:ratio-expansion} with $A=D_{xy}+2$, $B=D_x$, and $b_n=n\bar p$ to obtain
\[
\kappa(x,y)\ \le\ \frac{\mu_{xy}}{d_0}\ \pm\ C\,\varepsilon_n
\qquad\text{on }\mathcal E_{\deg}\cap\mathcal E_{\codeg},
\]
using $\bar p\le 1$ to absorb the co-degree fluctuation into the same $\varepsilon_n$ scale.

\smallskip\noindent
\emph{Lower bound via a matching.}
Assume $\{x,y\}\in E$ and relabel if necessary so that $D_x\le D_y$.
Define $U_x,U_y$ as in Lemma~\ref{lem:coupling}. Conditional on $\Gamma(x),\Gamma(y),\sigma$,
the edges in $U_x\times U_y$ are independent and have success probabilities in $\{p_0,p_1\}\ge p_1$.
Hence the conditional law stochastically dominates an i.i.d.\ $\mathrm{Ber}(p_1)$ bipartite graph on
$U_x\cup U_y$.

On $\mathcal E_{\deg}\cap\mathcal E_{\codeg}$, Assumption~\ref{ass:BC} implies $D_x,D_y\asymp n\bar p$ and
$D_{xy}\asymp n\bar p^2$. Moreover, for either community relation between $x$ and $y$,
a direct computation from Lemma~\ref{lem:codegree-concen} gives
$d_0-\mu_{xy}\asymp_\rho n\bar p$ (this uses $p_0\le 1-\rho$), hence
\[
|U_x|=D_x-1-D_{xy}\ \asymp_\rho\ n\bar p,
\qquad
|U_y|=D_y-1-D_{xy}\ \asymp_\rho\ n\bar p,
\]
and in particular $|U_y|\le \Lambda(\rho)|U_x|$ for some $\Lambda(\rho)<\infty$.
Indeed, using Lemma~\ref{lem:codegree-concen},
if $\sigma(x)=\sigma(y)$ then
\[
d_0-\mu_{xy}
=(n-1)p_0+np_1-\big((n-2)p_0^2+np_1^2\big)
=(n-2)p_0(1-p_0)+p_0+n p_1(1-p_1),
\]
and if $\sigma(x)\neq\sigma(y)$ then
\[
d_0-\mu_{xy}
=(n-1)p_0+np_1-(2n-2)p_0p_1
=(n-1)p_0(1-p_1)+n p_1(1-p_0)+p_0p_1.
\]
Under Assumption~\ref{ass:BC} (in particular $p_0\le 1-\rho$), both are $\asymp_\rho n\bar p$.

Since $p_1\asymp_\rho \bar p$ and $|U_y|\asymp n\bar p$, we have $p_1|U_y|\asymp_\rho n\bar p^2$.
Choosing $c_0$ large enough, Lemma~\ref{lem:bip-matching} (with $q=10$) yields that, conditional on
$\Gamma(x),\Gamma(y),\sigma$, with probability at least $1-|U_x|^{-10}$ there exists an injection
$\phi:U_x\to U_y$ such that $\{u,\phi(u)\}\in E$ for all $u\in U_x$, hence $d(u,\phi(u))=1$.

On this matching event, Lemma~\ref{lem:coupling} gives
\[
\kappa(x,y)\ \ge\ \frac{D_{xy}}{D_y}\ -\ 2\frac{D_y-D_x}{D_y}\ +\ \frac{2}{D_y}.
\]
Applying Lemma~\ref{lem:ratio-expansion} to $D_{xy}/D_y$ with $b_n=n\bar p$ yields
$\frac{D_{xy}}{D_y}=\frac{\mu_{xy}}{d_0}\pm C\varepsilon_n$ on $\mathcal E_{\deg}\cap\mathcal E_{\codeg}$, while
$\frac{D_y-D_x}{D_y}\lesssim \varepsilon_n$ on $\mathcal E_{\deg}$ and $1/D_y\lesssim (n\bar p)^{-1}=o(\varepsilon_n)$.
Thus
\[
\kappa(x,y)\ \ge\ \frac{\mu_{xy}}{d_0}\ -\ C\,\varepsilon_n
\qquad\text{on }\mathcal E_{\deg}\cap\mathcal E_{\codeg}\cap \mathcal E_{\match}(x,y),
\]
where $\mathcal E_{\match}(x,y)$ denotes the matching event above.

\smallskip
\noindent\emph{Uniformization over edges.}
On $\mathcal E_{\deg}\cap\mathcal E_{\codeg}$ we have $|U_x|\ge c\,n\bar p$, hence
$\Pr(\mathcal E_{\match}(x,y)^c\mid \Gamma(x),\Gamma(y),\sigma)\le (c\,n\bar p)^{-10}$.
Taking expectations and union bounding over all $\binom{2n}{2}\le 2n^2$ unordered pairs yields
\[
\Pr\Big(\exists\,\{x,y\}\in E:\ \mathcal E_{\match}(x,y)^c\Big)\ \le\ 2n^2\,(c\,n\bar p)^{-10}.
\]
Combining with $\Pr((\mathcal E_{\deg}\cap\mathcal E_{\codeg})^c)\le 2n^{-8}$ gives an overall failure probability
\[
\le 2n^{-8} + 2n^2\,(c\,n\bar p)^{-10} \;=\; o(1).
\]
In particular, for all sufficiently large $n$ this is $\le n^{-2}$.

Finally, $\mu_{xy}/d_0$ equals $w_{\rm in}^{(n)}$ when $\sigma(x)=\sigma(y)$ and $w_{\rm out}^{(n)}$ otherwise.
The fixed-pair refinement follows by the same argument without uniform union bounds.
\end{proof}

\begin{lemma}[Uniform positivity of curvature weights]\label{lem:curv-positive}
Assume Assumptions~\ref{ass:BC} and~\ref{ass:MDT}. Then there exists a constant $c=c(\rho)>0$ such that
\[
\Pr\Big(\min_{\{x,y\}\in E(G)} \kappa(x,y)\ \ge\ c\,\bar p\Big)\ \longrightarrow\ 1.
\]
In particular, the one-step Ricci reweighting $W^{(1)}_{xy}:=\kappa(x,y)\Ind\{\{x,y\}\in E\}$
produces nonnegative edge weights w.h.p., and its weighted degrees satisfy
\[
\min_{x\in[2n]} d_x(W^{(1)})\ \ge\ c\,n\bar p^{\,2}
\qquad\text{w.h.p.}
\]
\end{lemma}

\begin{proof}
By Lemma~\ref{lem:winwout-scale} there exists $c_0=c_0(\rho)>0$ such that for all sufficiently large $n$,
$\min\{w^{(n)}_{\rm in},w^{(n)}_{\rm out}\}\ge c_0\bar p$.
Let $\mathcal E$ be the event from Theorem~\ref{thm:sbm-main} on which, simultaneously for all edges,
$\big|\kappa(x,y)-w^{(n)}_{\rm type(x,y)}\big|\le C\varepsilon_n$.
Since $\varepsilon_n/\bar p=\eta_n\to 0$ under Assumption~\ref{ass:MDT}, for large $n$ we have
$C\varepsilon_n\le \frac12 c_0\bar p$, hence on $\mathcal E$,
$\kappa(x,y)\ge \frac12 c_0\bar p$ for all $\{x,y\}\in E$. This proves the first claim.

For the degree bound, on $\mathcal E$ we have $\kappa(x,y)\ge c\bar p$ on every edge, hence
$d_x(W^{(1)})=\sum_{y:\{x,y\}\in E}\kappa(x,y)\ge c\bar p\,D_x$.
Under Assumption~\ref{ass:MDT}, $n\bar p\gg \log n$, so Lemma~\ref{lem:degree-concen} gives
$\min_x D_x\asymp n\bar p$ w.h.p., hence $\min_x d_x(W^{(1)})\gtrsim n\bar p^2$ w.h.p.
\end{proof}

\section{One--step Ricci re--weighting and spectral clustering}\label{sec:one-step}

\paragraph{Standing model and notation.}
We work in the balanced two-block SBM on $V=[2n]$ with fixed labels
$\sigma:[2n]\to\{1,2\}$, blocks $B_a:=\sigma^{-1}(a)$, and $|B_1|=|B_2|=n$.
Write $\bar p:=\tfrac12(p_0+p_1)$ and
\[
r\ :=\ \frac{p_0-p_1}{p_0+p_1}\ \in(0,1),
\]
noting that $p_0=p_0(n)$ and $p_1=p_1(n)$ may depend on $n$.
Under Assumption~\ref{ass:BC}, $p_1/p_0\in[\rho,1-\rho]$ implies
\[
\frac{\rho}{2-\rho}\ \le\ r\ \le\ \frac{1-\rho}{1+\rho},
\]
so $r$ is bounded away from $0$ and $1$ uniformly in $n$.

Let $\ones\in\mathbb R^{2n}$ be the all--ones vector and set
\[
f\ :=\ \frac{1}{\sqrt{2n}}(\ones_n,-\ones_n)^\top\in\mathbb R^{2n}.
\]
We write $\|\cdot\|_{\op}$ for operator norm and $\|\cdot\|_{\max}$ for entrywise max norm.
Eigenvalues of a symmetric matrix are ordered increasingly: $\lambda_1\le\lambda_2\le\cdots$.

\paragraph{Sample objects.}
Let $A\in\{0,1\}^{2n\times 2n}$ be the (hollow, symmetric) adjacency matrix of $G$, so $A_{ii}=0$.
Define unweighted degrees $D_i:=\sum_{j}A_{ij}$ and the degree matrix $\widehat D_0:=\Deg(A)$
(Section~\ref{sec:prelim}).
For each edge $\{i,j\}\in E$, let $\kappa(i,j)$ denote the (random) Lin--Lu--Yau curvature of $G$
(Section~\ref{sec:curv-sbm}). Define the one-step curvature-weighted adjacency matrix $\widehat W$ by
\[
\widehat W_{ij}:=\kappa(i,j)\,\mathbbm{1}\{\{i,j\}\in E\},
\qquad
\widehat W_{ii}=0,
\]
and its weighted degrees $\widehat d_{1,i}:=d_i(\widehat W)=\sum_j\widehat W_{ij}$ and weighted degree matrix
$\widehat D_1:=\Deg(\widehat W)$.
The unweighted and one-step curvature-weighted normalized Laplacians are
\[
\widehat L_0:=L(A)=I_{2n}-\widehat D_0^{-1/2}A\,\widehat D_0^{-1/2},
\qquad
\widehat L_1:=L(\widehat W)=I_{2n}-\widehat D_1^{-1/2}\widehat W\,\widehat D_1^{-1/2}.
\]

\paragraph{Remark (centering for sign rounding).}
For $L(W)=I-D(W)^{-1/2}WD(W)^{-1/2}$ the null eigenvector is $D(W)^{1/2}\ones$ (not $\ones$),
so a $\lambda_2$-eigenvector of $\widehat L_k$ need not satisfy $\widehat v_k\perp \ones$.
To apply Lemma~\ref{lem:sign-error-fixed}, we therefore center the eigenvector via the projector
$\Pi:=I_{2n}-\frac{1}{2n}\ones\ones^\top$ before taking signs.

\paragraph{Standing window for this section.}
Throughout Section~\ref{sec:one-step} we assume Assumptions~\ref{ass:BC} and~\ref{ass:MDT}.
In particular, $n\bar p\gg\log n$ and $n\bar p^{\,2}\gg\log n$ (Remark~\ref{rem:MDT-conseq}), so:
(i) $\min_i D_i\ge 1$ w.h.p., hence $\widehat L_0$ is well-defined; and
(ii) curvature is uniformly positive on edges w.h.p.\ (Lemma~\ref{lem:curv-positive}), hence
$\min_i \widehat d_{1,i}>0$ w.h.p.\ and $\widehat L_1$ is well-defined.

\subsection{Population proxies}\label{subsec:objects}

\paragraph{Population adjacency.}
Conditioning on $\sigma$, define $P^{\pop}:=\E[A\mid\sigma]$.
Then $\Deg(P^{\pop})=d_0 I_{2n}$ with
\[
d_0=\E[D_i\mid\sigma]=(n-1)p_0+np_1 \asymp_\rho n\bar p,
\]
and we define
\[
L_0^{\pop}:=I_{2n}-d_0^{-1}P^{\pop}.
\]

\paragraph{Population curvature levels and population weighted Laplacian.}
Recall $w_{\rm in}^{(n)},w_{\rm out}^{(n)}$ from \eqref{eq:w-in-out-n}. Introduce the large-$n$ proxies
\begin{equation}\label{eq:w-in-out-pop-proxy-sec4-fixed}
w_{\rm in}^{\pop}:=\frac{p_0^2+p_1^2}{p_0+p_1}=\bar p(1+r^2),
\qquad
w_{\rm out}^{\pop}:=\frac{2p_0p_1}{p_0+p_1}=\bar p(1-r^2),
\end{equation}
so that $|w_{\rm in}^{(n)}-w_{\rm in}^{\pop}|\lesssim_\rho \bar p/n$ and
$|w_{\rm out}^{(n)}-w_{\rm out}^{\pop}|\lesssim_\rho \bar p/n$.

Let $K^{(n)}$ be the block-constant mask with $K^{(n)}_{ij}=w_{\rm in}^{(n)}$ if $\sigma(i)=\sigma(j)$,
$K^{(n)}_{ij}=w_{\rm out}^{(n)}$ if $\sigma(i)\neq\sigma(j)$, and $K^{(n)}_{ii}=0$.
Define the population weighted adjacency
\[
W_1^{\pop}:=K^{(n)}\circ P^{\pop},
\qquad
d_1:=(n-1)p_0w_{\rm in}^{(n)}+np_1w_{\rm out}^{(n)}.
\]
Then $\Deg(W_1^{\pop})=d_1 I_{2n}$ and we set
\[
L_1^{\pop}:=I_{2n}-d_1^{-1}W_1^{\pop}.
\]
Under Assumption~\ref{ass:BC}, Lemma~\ref{lem:winwout-scale} implies $d_1\asymp_\rho n\bar p^{\,2}$.

\subsection{Curvature sandwich and positivity}\label{subsec:curv-sandwich}

\begin{corollary}[Uniform curvature sandwich + positivity]\label{cor:curv-sandwich-fixed}
Assume Assumptions~\ref{ass:BC} and~\ref{ass:MDT}. Then there exist constants $c,C>0$ (depending only on $\rho$)
such that with probability at least $1-n^{-2}$, uniformly over all edges $\{i,j\}\in E$,
\begin{equation}\label{eq:curv-sandwich-sec4-fixed}
\kappa(i,j)=
\begin{cases}
w_{\rm in}^{(n)}\ \pm\ C\,\varepsilon_n,& \sigma(i)=\sigma(j),\\[3pt]
w_{\rm out}^{(n)}\ \pm\ C\,\varepsilon_n,& \sigma(i)\neq \sigma(j),
\end{cases}
\qquad
\varepsilon_n=\sqrt{\frac{\log n}{n\bar p}},
\end{equation}
and moreover on the same event,
\begin{equation}\label{eq:curv-positive-sec4-fixed}
\min_{\{i,j\}\in E}\kappa(i,j)\ \ge\ c\,\bar p,
\qquad\text{and}\qquad
\|\widehat W\|_{\max}\ \le\ C\,\bar p.
\end{equation}
In particular, $\widehat D_1$ has strictly positive diagonal entries and $\widehat L_1$ is well-defined.
\end{corollary}

\begin{proof}
The sandwich \eqref{eq:curv-sandwich-sec4-fixed} is Theorem~\ref{thm:sbm-main}.
Under Assumption~\ref{ass:MDT}, $\varepsilon_n=o(\bar p)$, so combining
Lemma~\ref{lem:winwout-scale} with \eqref{eq:curv-sandwich-sec4-fixed} yields
$\min_{\{i,j\}\in E}\kappa(i,j)\ge c\bar p$ for all large $n$.
Finally, $\widehat W_{ij}=\kappa(i,j)\mathbbm{1}\{\{i,j\}\in E\}$ implies
$\|\widehat W\|_{\max}\le \max_{\{i,j\}\in E}\kappa(i,j)\lesssim \bar p$ on the same event.
\end{proof}

\subsection{Population spectrum and mean--field gain}\label{subsec:pop-spectrum}

\begin{lemma}[Population spectrum from block form]\label{lem:block-spectrum-fixed}
Fix $(w_{\rm in},w_{\rm out})$ and set $a:=p_0w_{\rm in}$, $b:=p_1w_{\rm out}$,
and $d_*:=(n-1)a+nb$.
Let
\[
W=
\begin{pmatrix}
a\,(J_n-I_n) & b\,J_n\\
b\,J_n & a\,(J_n-I_n)
\end{pmatrix},
\qquad
L=I_{2n}-\frac{1}{d_*}W.
\]
Then $\Deg(W)=d_*I_{2n}$, $\lambda_1(L)=0$ with eigenvector $\ones$, and:
\begin{itemize}
\item the $2(n-1)$ within--block zero--sum directions have eigenvalue $1+a/d_*$;
\item in the block--constant subspace, the remaining eigenpair is $(\lambda_2(L),f)$ with
\[
\lambda_2(L)=1-r_n(w_{\rm in},w_{\rm out}),
\qquad
r_n(w_{\rm in},w_{\rm out})
:=\frac{(n-1)p_0w_{\rm in}-np_1w_{\rm out}}{(n-1)p_0w_{\rm in}+np_1w_{\rm out}}.
\]
\end{itemize}
\end{lemma}

\begin{proof}
Let $u$ be supported on $B_1$ with $\sum_{i\in B_1}u_i=0$ (the other block case is identical).
Then $J_n u_{B_1}=0$ and $(J_n-I_n)u_{B_1}=-u_{B_1}$, while the off-diagonal block annihilates $u$.
Hence $Wu=-a u$ and $Lu=(1+a/d_*)u$.

For block-constant vectors $x=(c_1\ones_n,c_2\ones_n)$,
\[
Wx=\big(((n-1)a\,c_1+nb\,c_2)\ones_n,\ (nb\,c_1+(n-1)a\,c_2)\ones_n\big),
\]
so on the basis $\{(\ones_n,0),(0,\ones_n)\}$ the restriction of $W$ equals
$\begin{psmallmatrix}(n-1)a & nb\\ nb & (n-1)a\end{psmallmatrix}$,
whose eigenvectors are $\ones$ and $f$ with eigenvalues $d_*=(n-1)a+nb$ and $(n-1)a-nb$.
Normalizing by $d_*$ gives the stated eigenvalues of $L$.
\end{proof}

\begin{lemma}[Mean--field eigengap gain and curvature contrast]\label{lem:mf-gain-fixed}
Let $\Gamma_k^{\pop}:=\lambda_3(L_k^{\pop})-\lambda_2(L_k^{\pop})$ for $k\in\{0,1\}$ and define
\[
r_n^{(0)}:=r_n(1,1)
=\frac{(n-1)p_0-np_1}{(n-1)p_0+np_1},
\qquad
r_n^{(1)}:=r_n\!\big(w_{\rm in}^{(n)},w_{\rm out}^{(n)}\big).
\]
Then, uniformly under Assumption~\ref{ass:BC},
\begin{equation}\label{eq:mf-gap-fixed}
\Gamma_k^{\pop}
=
r_n^{(k)}\ +\ O\!\Big(\frac{1}{n}\Big),
\qquad k\in\{0,1\}.
\end{equation}
Moreover, writing $r=(p_0-p_1)/(p_0+p_1)$ and defining
\begin{equation}\label{eq:r-curv-def-fixed}
r_{\rm curv}:=
\frac{p_0 w_{\rm in}^{\pop}-p_1 w_{\rm out}^{\pop}}{p_0 w_{\rm in}^{\pop}+p_1 w_{\rm out}^{\pop}}
=\frac{r}{1-r+r^2},
\end{equation}
we have
\begin{equation}\label{eq:mf-expansions-fixed}
r_n^{(0)}=r+O\!\Big(\frac{1}{n}\Big),
\qquad
r_n^{(1)}=r_{\rm curv}+O\!\Big(\frac{1}{n}\Big),
\end{equation}
and hence
\begin{equation}\label{eq:mf-gain-fixed}
\Gamma_1^{\pop}-\Gamma_0^{\pop}
=
\big(r_{\rm curv}-r\big)+O\!\Big(\frac{1}{n}\Big)
=
\frac{r^2(1-r)}{1-r+r^2}+O\!\Big(\frac{1}{n}\Big).
\end{equation}
Under Assumption~\ref{ass:BC}, the leading term $r_{\rm curv}-r$ is bounded below by a positive constant depending only on $\rho$.
\end{lemma}

\begin{proof}
Apply Lemma~\ref{lem:block-spectrum-fixed}. For $L_0^{\pop}$ take $(w_{\rm in},w_{\rm out})=(1,1)$.
For $L_1^{\pop}$ take $(w_{\rm in},w_{\rm out})=(w_{\rm in}^{(n)},w_{\rm out}^{(n)})$.
In either case, $\lambda_3(L)=1+a/d_*$ with $a/d_*=O(1/n)$ uniformly, while
$\lambda_2(L)=1-r_n(\cdot,\cdot)$, which gives \eqref{eq:mf-gap-fixed}.
The expansions \eqref{eq:mf-expansions-fixed} follow by (i) replacing $n-1$ by $n$ at cost $O(1/n)$ and
(ii) replacing $(w_{\rm in}^{(n)},w_{\rm out}^{(n)})$ by $(w_{\rm in}^{\pop},w_{\rm out}^{\pop})$ at cost $O(1/n)$
using \eqref{eq:w-in-out-pop-proxy-sec4-fixed}. The identity in \eqref{eq:r-curv-def-fixed} is algebraic.
Finally, under Assumption~\ref{ass:BC}, $r$ ranges over a compact subset of $(0,1)$, so
$\inf (r_{\rm curv}-r)>0$ over that range.
\end{proof}

\subsection{Auxiliary perturbation lemmas}\label{subsec:aux-perturb}

\begin{lemma}[Row/column sum bound]\label{lem:op-1inf-fixed}
For any real matrix $M$,
\[
\|M\|_{\op}\ \le\ \sqrt{\|M\|_{1}\,\|M\|_{\infty}},
\]
where $\|M\|_{\infty}:=\max_i\sum_j |M_{ij}|$ and $\|M\|_{1}:=\max_j\sum_i |M_{ij}|$.
In particular, if $M$ is symmetric then $\|M\|_{\op}\le \|M\|_{\infty}$.
\end{lemma}

\begin{proof}
Standard: $\|M\|_{\op}^2=\|M^\top M\|_{\op}\le \|M^\top M\|_1\le \|M^\top\|_1\|M\|_1=\|M\|_\infty\|M\|_1$.
\end{proof}

\begin{lemma}[Inverse square-root perturbation on diagonals]\label{lem:inv-sqrt-diag-fixed}
Let $A=\diag(a)$ and $B=\diag(b)$ with $a_i,b_i\ge m>0$ for all $i$.
Then
\[
\|A^{-1/2}-B^{-1/2}\|_{\op}
\ \le\ \tfrac12\,m^{-3/2}\,\|a-b\|_\infty.
\]
\end{lemma}

\begin{proof}
Apply the scalar mean value theorem to $\varphi(t)=t^{-1/2}$ coordinatewise.
\end{proof}

\begin{lemma}[From principal angle to sign error]\label{lem:sign-error-fixed}
Let $f=(\ones_n,-\ones_n)^\top/\sqrt{2n}$ and let $v\perp\ones$ be unit.
Let $\widehat\sigma_v$ be the sign clustering induced by $v$, with misclassification rate
\[
\mathrm{err}(\widehat\sigma_v)
:=\frac{1}{2n}\min\Big\{\big|\{i:\widehat\sigma_v(i)\neq \sigma(i)\}\big|,\ \big|\{i:\widehat\sigma_v(i)= \sigma(i)\}\big|\Big\}.
\]
Then
\[
\mathrm{err}(\widehat\sigma_v)\ \le\ \min\Big\{\frac12,\ \tan^2\angle(v,f)\Big\}.
\]
\end{lemma}

\begin{proof}
Choose the global sign of $v$ so that $\langle v,f\rangle\ge 0$ and write
$v=\alpha f+g$ with $g\perp f$, $\alpha=\cos\theta$, $\|g\|_2=\sin\theta$, $\theta=\angle(v,f)\in[0,\pi/2]$.
Let $M$ be the misclassified set (after the optimal global flip). For $i\in M$, $v_i$ and $f_i$ have opposite signs,
so $|v_i-\alpha f_i|\ge \alpha|f_i|=\alpha/\sqrt{2n}$. Hence
\[
\sin^2\theta=\|g\|_2^2=\|v-\alpha f\|_2^2
\ \ge\ \sum_{i\in M}(v_i-\alpha f_i)^2
\ \ge\ |M|\cdot \frac{\alpha^2}{2n},
\]
so $\mathrm{err}(\widehat\sigma_v)=|M|/(2n)\le \sin^2\theta/\alpha^2=\tan^2\theta$.
\end{proof}

\subsection{Degree and spectral controls}\label{subsec:deg-spec}

\begin{lemma}[Degree, weighted-degree, and adjacency spectral bounds]\label{lem:deg-spec-fixed}
Assume Assumptions~\ref{ass:BC} and~\ref{ass:MDT}. There exist constants $c,C>0$ such that with probability
at least $1-Cn^{-6}$,
\begin{align}
\min_i D_i &\ge c\,n\bar p,
&\max_i |D_i-d_0| &\le C\sqrt{n\bar p\log n},
\label{eq:deg-bnds-fixed}\\[4pt]
\min_i \widehat d_{1,i} &\ge c\,d_1,
&\max_i |\widehat d_{1,i}-d_1| &\le C\sqrt{n\bar p\log n},
\label{eq:wdeg-bnds-fixed}\\[4pt]
\|A-P^{\pop}\|_{\op} &\le C\sqrt{n\bar p\log n},
\label{eq:A-pop-fixed}\\[4pt]
\|A\|_{\op} &\le C\,n\bar p.
\label{eq:A-op-fixed}
\end{align}
\end{lemma}

\begin{proof}
The degree bounds \eqref{eq:deg-bnds-fixed} follow from Lemma~\ref{lem:degree-concen}, since
Assumption~\ref{ass:MDT} implies $n\bar p\gg \log n$.

For \eqref{eq:A-pop-fixed}, apply matrix Bernstein conditionally on $\sigma$ to $A-P^{\pop}$.
Then \eqref{eq:A-op-fixed} follows from $\|A\|_{\op}\le \|A-P^{\pop}\|_{\op}+\|P^{\pop}\|_{\op}$ and
$\|P^{\pop}\|_{\op}\asymp n\bar p$.

For \eqref{eq:wdeg-bnds-fixed}, decompose for each $i$:
\[
\widehat d_{1,i}
=\sum_{j\ne i}A_{ij}\kappa(i,j)
=\sum_{j\ne i}A_{ij}K^{(n)}_{ij}+\sum_{j\ne i}A_{ij}\big(\kappa(i,j)-K^{(n)}_{ij}\big)
=:S_i+R_i.
\]
Conditionally on $\sigma$, $S_i$ is a sum of independent bounded variables with mean
$\E[S_i\mid\sigma]=d_1\asymp_\rho n\bar p^2$ and summands bounded by $\|K^{(n)}\|_{\max}\lesssim_\rho \bar p$,
so Bernstein and a union bound give $\max_i|S_i-d_1|\lesssim \sqrt{n\bar p^{\,3}\log n}+\bar p\log n$.
On the event of Corollary~\ref{cor:curv-sandwich-fixed}, $\|\kappa-K^{(n)}\|_{\max}\lesssim \varepsilon_n$, hence
\[
\max_i|R_i|
\le \|\kappa-K^{(n)}\|_{\max}\max_i D_i
\lesssim \varepsilon_n\,(n\bar p)
=\sqrt{n\bar p\log n},
\]
using \eqref{eq:deg-bnds-fixed}. Absorbing $\sqrt{n\bar p^{\,3}\log n}+\bar p\log n\le \sqrt{n\bar p\log n}$
(for $\bar p\le 1$ and $n\bar p\gg\log n$) yields the stated (looser) deviation bound in \eqref{eq:wdeg-bnds-fixed}.
Finally, $\sqrt{n\bar p\log n}=o(d_1)$ under Assumption~\ref{ass:MDT}, so $\min_i\widehat d_{1,i}\ge c d_1$
for all large $n$ on the same event.
\end{proof}

\subsection{Proxy decomposition and operator-norm concentration}\label{subsec:op-conc}

\begin{lemma}[Proxy decomposition]\label{lem:proxy-decomp-fixed}
Define the proxy Laplacian
\[
\widetilde L_1:=I_{2n}-d_1^{-1}\widehat W.
\]
Then
\begin{equation}\label{eq:proxy-id-fixed}
\widehat L_1-L^{\pop}_1
=(\widehat L_1-\widetilde L_1)
+d_1^{-1}\big(K^{(n)}\circ(A-P^{\pop})\big)
+d_1^{-1}\big(\widehat W-K^{(n)}\circ A\big),
\end{equation}
and the normalizer remainder satisfies
\[
\widehat L_1-\widetilde L_1
=\ -\Big(\widehat D_1^{-1/2}-d_1^{-1/2}I\Big)\widehat W\,\widehat D_1^{-1/2}
\ -\ d_1^{-1/2}\widehat W\Big(\widehat D_1^{-1/2}-d_1^{-1/2}I\Big).
\]
\end{lemma}

\begin{proof}
Add and subtract $d_1^{-1}\widehat W$ and expand
\[
\widehat W-W_1^{\pop}
=\widehat W-(K^{(n)}\circ P^{\pop})
=K^{(n)}\circ(A-P^{\pop})+\big(\widehat W-K^{(n)}\circ A\big).
\]
The second identity is the polarization expansion of
$\widehat D_1^{-1/2}\widehat W\widehat D_1^{-1/2}-d_1^{-1}\widehat W$.
\end{proof}

\begin{lemma}[Operator--norm concentration for $\widehat L_0$ and $\widehat L_1$]\label{lem:conc-L-fixed}
Assume Assumptions~\ref{ass:BC} and~\ref{ass:MDT}. There exist constants $C>0$ such that with probability
at least $1-Cn^{-6}$,
\begin{equation}\label{eq:delta0-fixed}
\delta_0:=\|\widehat L_0-L^{\pop}_0\|_{\op}\ \le\ C\,\varepsilon_n,
\qquad
\varepsilon_n=\sqrt{\frac{\log n}{n\bar p}},
\end{equation}
and
\begin{equation}\label{eq:delta1-fixed}
\delta_1:=\|\widehat L_1-L^{\pop}_1\|_{\op}\ \le\ C\big(\varepsilon_n+\eta_n\big),
\qquad
\eta_n:=\sqrt{\frac{\log n}{n\bar p^{\,3}}}.
\end{equation}
\end{lemma}

\begin{proof}
Work on the high-probability event of Lemma~\ref{lem:deg-spec-fixed} and
Corollary~\ref{cor:curv-sandwich-fixed}.

\smallskip\noindent
\emph{Step 1: bound $\delta_0$.}
Write
\[
\widehat L_0-L_0^{\pop}
=\Big(\widehat D_0^{-1/2}-d_0^{-1/2}I\Big)A\,\widehat D_0^{-1/2}
+d_0^{-1/2}(A-P^{\pop})d_0^{-1/2}
+d_0^{-1/2}P^{\pop}\Big(\widehat D_0^{-1/2}-d_0^{-1/2}I\Big).
\]
Since $d_0\asymp n\bar p$ and $\max_i|D_i-d_0|\lesssim \sqrt{n\bar p\log n}$,
\[
\|\widehat D_0^{-1/2}-d_0^{-1/2}I\|_{\op}
=\max_i\Big|D_i^{-1/2}-d_0^{-1/2}\Big|
\lesssim d_0^{-3/2}\max_i|D_i-d_0|
\lesssim \frac{\varepsilon_n}{\sqrt{n\bar p}}.
\]
Moreover, $\|A\|_{\op}\lesssim n\bar p$ and $\|A-P^{\pop}\|_{\op}\lesssim \sqrt{n\bar p\log n}$ by
Lemma~\ref{lem:deg-spec-fixed}, and $\|P^{\pop}\|_{\op}\asymp n\bar p$.
Combining these bounds yields $\|\widehat L_0-L_0^{\pop}\|_{\op}\lesssim \varepsilon_n$.

\smallskip\noindent
\emph{Step 2: bound $\delta_1$ via Lemma~\ref{lem:proxy-decomp-fixed}.}
Use \eqref{eq:proxy-id-fixed}.

\smallskip\noindent
\emph{(i) Independent core term.}
Let $M:=K^{(n)}\circ(A-P^{\pop})$. Conditionally on $\sigma$, matrix Bernstein gives
$\|M\|_{\op}\lesssim \sqrt{n\bar p^{\,3}\log n}+\bar p\log n\lesssim \sqrt{n\bar p^{\,3}\log n}$,
using $n\bar p\gg\log n$ and $\|K^{(n)}\|_{\max}\lesssim \bar p$.
Since $d_1\asymp n\bar p^2$,
\[
\|d_1^{-1}M\|_{\op}
\lesssim \frac{\sqrt{n\bar p^{\,3}\log n}}{n\bar p^{\,2}}
=\varepsilon_n.
\]

\smallskip\noindent
\emph{(ii) Curvature fluctuation term.}
Entrywise, $\widehat W-K^{(n)}\circ A = (\kappa-K^{(n)})\circ A$. Thus by Lemma~\ref{lem:op-1inf-fixed},
\[
\|\widehat W-K^{(n)}\circ A\|_{\op}
\le \max_i \sum_j A_{ij}\,|\kappa(i,j)-K^{(n)}_{ij}|
\le \|\kappa-K^{(n)}\|_{\max}\,\max_i D_i
\lesssim \varepsilon_n\,(n\bar p)=\sqrt{n\bar p\log n}.
\]
Therefore
\[
\|d_1^{-1}(\widehat W-K^{(n)}\circ A)\|_{\op}
\lesssim \frac{\sqrt{n\bar p\log n}}{n\bar p^{\,2}}
=\eta_n.
\]

\smallskip\noindent
\emph{(iii) Degree-normalizer remainder.}
By Lemma~\ref{lem:deg-spec-fixed}, $\min_i \widehat d_{1,i}\asymp d_1$ and
$\max_i|\widehat d_{1,i}-d_1|\lesssim \sqrt{n\bar p\log n}$, so Lemma~\ref{lem:inv-sqrt-diag-fixed} yields
\[
\big\|\widehat D_1^{-1/2}-d_1^{-1/2}I\big\|_{\op}
\lesssim \frac{\sqrt{n\bar p\log n}}{d_1^{3/2}}
\asymp \frac{\sqrt{\log n}}{n\,\bar p^{5/2}}.
\]
Also $\|\widehat D_1^{-1/2}\|_{\op}\lesssim d_1^{-1/2}\asymp (n\bar p^2)^{-1/2}$ and
$\|\widehat W\|_{\op}\le \|\widehat W\|_{\infty}=\max_i \widehat d_{1,i}\lesssim n\bar p^2$.
Plugging into Lemma~\ref{lem:proxy-decomp-fixed} gives
$\|\widehat L_1-\widetilde L_1\|_{\op}\lesssim \eta_n$.

Combining (i)--(iii) yields \eqref{eq:delta1-fixed}.
\end{proof}

\subsection{Eigengap gain and misclassification bounds}\label{subsec:snr}

\begin{theorem}[One--step Ricci: gap gain and error control]\label{thm:main-gap-error-fixed}
Assume Assumptions~\ref{ass:BC} and~\ref{ass:MDT}. Let
\[
\Gamma_k:=\lambda_3(\widehat L_k)-\lambda_2(\widehat L_k),
\qquad
\Gamma_k^{\pop}:=\lambda_3(L_k^{\pop})-\lambda_2(L_k^{\pop}),
\qquad k\in\{0,1\}.
\]
There exists $C>0$ such that with probability at least $1-Cn^{-6}$:

\smallskip
\noindent\textbf{(i) Sample eigengap comparison.}
\begin{equation}\label{eq:gap-comp-fixed}
\Gamma_1-\Gamma_0
\ \ge\ (\Gamma_1^{\pop}-\Gamma_0^{\pop})\ -\ 2(\delta_0+\delta_1),
\end{equation}
and hence, by Lemmas~\ref{lem:mf-gain-fixed} and~\ref{lem:conc-L-fixed},
\begin{equation}\label{eq:gap-comp-explicit-fixed}
\Gamma_1-\Gamma_0
\ \ge\ \big(r_{\rm curv}-r\big)\ -\ C\big(\varepsilon_n+\eta_n\big)\ -\ Cn^{-1}.
\end{equation}

\smallskip
\smallskip
\noindent\textbf{(ii) Misclassification via Davis--Kahan.}
Let $\widehat v_k$ be a unit eigenvector of $\widehat L_k$ for $\lambda_2(\widehat L_k)$ (unique up to sign).
Define the centering projector and centered unit vector
\[
\Pi:=I_{2n}-\frac{1}{2n}\ones\ones^\top,
\qquad
\widetilde v_k:=\frac{\Pi \widehat v_k}{\|\Pi \widehat v_k\|_2},
\]
and let $\widehat\sigma_k$ be the sign clustering from $\widetilde v_k$.
If $\delta_k\le \tfrac14\Gamma_k^{\pop}$, then
\begin{equation}\label{eq:err-bnd-fixed}
\mathrm{err}(\widehat\sigma_k)
\ \le\ C\left(\frac{\delta_k}{\Gamma_k^{\pop}}\right)^2,
\qquad k\in\{0,1\}.
\end{equation}
\end{theorem}

\begin{proof}
For (i), Weyl's inequality gives $|\lambda_j(\widehat L_k)-\lambda_j(L_k^{\pop})|\le \delta_k$ for $j=2,3$,
hence $\Gamma_k\ge \Gamma_k^{\pop}-2\delta_k$ and \eqref{eq:gap-comp-fixed} follows. Insert
\eqref{eq:mf-gain-fixed} and \eqref{eq:delta0-fixed}--\eqref{eq:delta1-fixed} to obtain \eqref{eq:gap-comp-explicit-fixed}.

For (ii), Lemma~\ref{lem:block-spectrum-fixed} implies that $f$ is an eigenvector of $L_k^{\pop}$
for $\lambda_2(L_k^{\pop})$ and that
\[
\lambda_1(L_k^{\pop})=0,\qquad
\Gamma_k^{\pop}=\lambda_3(L_k^{\pop})-\lambda_2(L_k^{\pop}).
\]
Moreover, still by Lemma~\ref{lem:block-spectrum-fixed},
\[
\lambda_2(L_k^{\pop})=1-r_n^{(k)}+O(1/n),\qquad \Gamma_k^{\pop}=r_n^{(k)}+O(1/n),
\]
and under Assumption~\ref{ass:BC} we have $r_n^{(k)}\in[c(\rho),1-c(\rho)]$ for all large $n$,
so the separation of $\lambda_2(L_k^{\pop})$ from the rest of the spectrum satisfies
\[
\sep_k^{\pop}:=\min\{\lambda_2(L_k^{\pop})-\lambda_1(L_k^{\pop}),\ \lambda_3(L_k^{\pop})-\lambda_2(L_k^{\pop})\}
=\min\{\lambda_2(L_k^{\pop}),\Gamma_k^{\pop}\}\ \ge\ c'(\rho)\,\Gamma_k^{\pop}.
\]
Therefore Davis--Kahan yields
\[
\sin\angle(\widehat v_k,f)\le C\,\frac{\delta_k}{\sep_k^{\pop}}
\le C(\rho)\,\frac{\delta_k}{\Gamma_k^{\pop}}.
\]
Assume $\delta_k\le \Gamma_k^{\pop}/4$, so $\angle(\widehat v_k,f)\le \pi/6$ for all large $n$,
and hence $\tan^2\angle(\widehat v_k,f)\lesssim (\delta_k/\Gamma_k^{\pop})^2$.

Next, set $\Pi:=I_{2n}-\frac{1}{2n}\ones\ones^\top$ and $\widetilde v_k=(\Pi\widehat v_k)/\|\Pi\widehat v_k\|_2$.
Since $\widehat v_k$ is not parallel to $\ones$ (because $\lambda_2(\widehat L_k)>0$), we have $\Pi\widehat v_k\neq 0$.
Choose the global sign of $\widehat v_k$ so that $\langle \widehat v_k,f\rangle\ge 0$ and write
$\widehat v_k=\alpha f+g$ with $g\perp f$, $\alpha=\cos\theta$, $\|g\|_2=\sin\theta$,
$\theta=\angle(\widehat v_k,f)\in[0,\pi/2]$.
Since $f\perp \ones$, we have $\Pi f=f$, and by symmetry of $\Pi$,
$\langle \Pi g,f\rangle=\langle g,\Pi f\rangle=\langle g,f\rangle=0$.
Thus
\[
\Pi\widehat v_k=\alpha f+\Pi g,\qquad \|\Pi g\|_2\le \|g\|_2=\sin\theta,
\qquad \|\Pi\widehat v_k\|_2^2=\alpha^2+\|\Pi g\|_2^2\ge \alpha^2,
\]
which implies
\[
\tan\angle(\widetilde v_k,f)=\frac{\|\Pi g\|_2}{\alpha}\le \frac{\sin\theta}{\cos\theta}
=\tan\angle(\widehat v_k,f).
\]
Finally, $\widetilde v_k\perp \ones$ by construction, so Lemma~\ref{lem:sign-error-fixed} applied to $\widetilde v_k$ gives
\[
\mathrm{err}(\widehat\sigma_k)\le \tan^2\angle(\widetilde v_k,f)\le \tan^2\angle(\widehat v_k,f)
\lesssim \left(\frac{\delta_k}{\Gamma_k^{\pop}}\right)^2,
\]
which is \eqref{eq:err-bnd-fixed}.
\end{proof}

\begin{corollary}[Strict eigengap improvement under Assumption~\ref{ass:MDT}]\label{cor:snr-gain-fixed}
Assume Assumptions~\ref{ass:BC} and~\ref{ass:MDT}.
Then for all sufficiently large $n$, with probability at least $1-Cn^{-6}$,
\[
\lambda_3(\widehat L_1)-\lambda_2(\widehat L_1)\ >\ \lambda_3(\widehat L_0)-\lambda_2(\widehat L_0).
\]
\end{corollary}

\begin{proof}
By Lemma~\ref{lem:mf-gain-fixed} and Assumption~\ref{ass:BC}, $(r_{\rm curv}-r)\ge c(\rho)>0$.
Under Assumption~\ref{ass:MDT}, $\varepsilon_n+\eta_n\to 0$, so \eqref{eq:gap-comp-explicit-fixed} implies
$\Gamma_1-\Gamma_0>0$ for all large $n$ on the stated high-probability event.
\end{proof}

\begin{corollary}[When one-step reweighting improves the Davis--Kahan misclassification bound]\label{cor:miscl-improve-fixed}
Assume Assumptions~\ref{ass:BC} and~\ref{ass:MDT}. On the event where \eqref{eq:err-bnd-fixed} holds for both
$k\in\{0,1\}$ (as in Theorem~\ref{thm:main-gap-error-fixed}), we have
\[
\frac{\mathrm{err}(\widehat\sigma_1)}{\mathrm{err}(\widehat\sigma_0)}
\ \le\
C\Big(\frac{\delta_1}{\delta_0}\Big)^2\Big(\frac{\Gamma^{\pop}_0}{\Gamma^{\pop}_1}\Big)^2.
\]
In particular, if
\begin{equation}\label{eq:suff-improve-fixed}
\frac{\delta_1}{\Gamma^{\pop}_1}\ \le\ \frac{\delta_0}{\Gamma^{\pop}_0},
\end{equation}
then $\mathrm{err}(\widehat\sigma_1)\le C\,\mathrm{err}(\widehat\sigma_0)$ on the same event.
\end{corollary}

\section{Finite horizon iterated Ricci reweighting}\label{sec:iterated-ricci}

We analyze a fixed horizon $T\in\mathbb N$ of curvature--driven edge reweighting in the balanced
two--block SBM. All Ollivier/LLY notions for weighted graphs are as in Section~\ref{sec:prelim}.
Throughout, all transport costs are computed in the \emph{unweighted} graph metric.

\paragraph{Standing assumptions and rates.}
Fix $T\in\mathbb N$. Throughout Section~\ref{sec:iterated-ricci} we assume Assumptions~\ref{ass:BC}
and~\ref{ass:MDT}. Recall $\bar p=(p_0+p_1)/2$ and define
\begin{equation}\label{eq:rates-eps-eta-iter-sec5}
\varepsilon_n := \sqrt{\frac{\log n}{n\bar p}},
\qquad
\eta_n := \frac{\varepsilon_n}{\bar p}=\sqrt{\frac{\log n}{n\bar p^3}}.
\end{equation}
Under Assumption~\ref{ass:MDT}, $\varepsilon_n=o(\bar p)$ and $\eta_n=o(1)$.

\paragraph{A finite-horizon strengthening for iterate tracking.}
To obtain a clean \emph{uniform-in-$t\le T$} iterate tracking bound with vanishing error, we will additionally
assume, when explicitly stated, the stronger condition
\begin{equation}\label{eq:MDT-T-sec5}
\tag{$\mathrm{MDT}(T)$}
n\,\bar p^{\,2T+1}\ \gg\ \log n,
\qquad\text{equivalently}\qquad
\eta_{n,T}:=\frac{\varepsilon_n}{\bar p^{\,T}}=\sqrt{\frac{\log n}{n\bar p^{\,2T+1}}}=o(1).
\end{equation}
In particular, \eqref{eq:MDT-T-sec5} implies $\varepsilon_n/\bar p^{T-1}=o(\bar p)$. Heuristically, the one-step curvature map is only Lipschitz at scale $1/\bar p$ because neighbor probabilities are
of order $1/(n\bar p)$ while weighted degrees are of order $n\bar p^2$. Thus, max-norm errors in weights are amplified
by a factor $\asymp 1/\bar p$ per iteration. Condition $(\mathrm{MDT}(T))$ ensures the amplified error remains $o(1)$
uniformly for $t\le T$.

\subsection{Iteration and a two-level benchmark recursion}\label{subsec:iter-benchmark}

Let $G\sim \mathrm{SBM}(2n,p_0,p_1,\sigma)$ on $V=[2n]$ with blocks $B_1,B_2$ of size $n$ and adjacency
matrix $A$. We iterate Ricci reweighting on the \emph{fixed} edge support:
\begin{equation}\label{eq:iterate-def-sec5}
W^{(0)} := A,
\qquad
W^{(t+1)} := \kappa_{W^{(t)}}\circ A
\quad (t\ge 0),
\end{equation}
where $(\kappa_W)_{xy}:=\kappa_W(x,y)$ for $\{x,y\}\in E$ and $0$ otherwise. Thus each $W^{(t)}$ is symmetric,
nonnegative, and supported on $E$.

\paragraph{Observed two-level benchmark.}
For $(a,b)\in(0,\infty)^2$ define
\begin{equation}\label{eq:Wstar-ab-sec5}
W^\star(a,b):=K(a,b)\circ A,
\qquad
K(a,b)_{xy}:=
\begin{cases}
a,& \sigma(x)=\sigma(y),\\
b,& \sigma(x)\neq\sigma(y),\\
0,& x=y.
\end{cases}
\end{equation}

\paragraph{Mean-field map and deterministic benchmark recursion.}
Define $\Phi_n:(0,\infty)^2\to(0,\infty)^2$ by
\begin{equation}\label{eq:Phi-def-sec5}
\Phi_n(a,b) := \bigl(\varphi_{\mathrm{in}}^{(n)}(a,b),\ \varphi_{\mathrm{out}}^{(n)}(a,b)\bigr),
\end{equation}
where
\begin{equation}\label{eq:phi-in-out-sec5}
\varphi_{\mathrm{in}}^{(n)}(a,b)
:= \frac{(n-2)p_0^2\,a + n p_1^2\,b}{(n-1)p_0\,a + n p_1\,b},
\qquad
\varphi_{\mathrm{out}}^{(n)}(a,b)
:= \frac{(2n-2)p_0p_1\,b}{(n-1)p_0\,a + n p_1\,b}.
\end{equation}
Set $w_{\mathrm{in}}^{(0)}=w_{\mathrm{out}}^{(0)}=1$ and define for $t\ge 0$
\begin{equation}\label{eq:w-recursion-sec5}
\bigl(w_{\mathrm{in}}^{(t+1)},w_{\mathrm{out}}^{(t+1)}\bigr)
:=\Phi_n\bigl(w_{\mathrm{in}}^{(t)},w_{\mathrm{out}}^{(t)}\bigr).
\end{equation}
Write $K^{(t)}:=K(w_{\mathrm{in}}^{(t)},w_{\mathrm{out}}^{(t)})$ and
\[
W^{\star,(t)} := W^\star\bigl(w_{\mathrm{in}}^{(t)},w_{\mathrm{out}}^{(t)}\bigr)=K^{(t)}\circ A.
\]

\begin{lemma}[Benchmark weights are $\asymp \bar p$ for $t\ge 1$]\label{lem:uniform-order-benchmark-sec5}
Fix $T\in\mathbb N$ and assume Assumption~\ref{ass:BC}. Then there exist constants
$0<c_{T,\rho}\le C_{T,\rho}<\infty$ such that for all $t\in\{1,\dots,T\}$,
\[
c_{T,\rho}\,\bar p \le w_{\mathrm{out}}^{(t)} \le w_{\mathrm{in}}^{(t)} \le C_{T,\rho}\,\bar p.
\]
\end{lemma}

\begin{proof}
For $t=1$, \eqref{eq:phi-in-out-sec5} gives $w_{\mathrm{in}}^{(1)},w_{\mathrm{out}}^{(1)}\asymp_\rho \bar p$ since
each numerator is $\Theta(n\bar p^2)$ and the denominator is $\Theta(n\bar p)$ under Assumption~\ref{ass:BC}.
If $0<b\le a$ and $a,b\asymp \bar p$, then the denominator $(n-1)p_0a+np_1b\asymp n\bar p^2$, and
$\varphi_{\mathrm{in}}^{(n)}(a,b),\varphi_{\mathrm{out}}^{(n)}(a,b)\asymp \bar p$ with
$\varphi_{\mathrm{in}}^{(n)}(a,b)\ge \varphi_{\mathrm{out}}^{(n)}(a,b)$ because
$p_0-p_1\gtrsim_\rho \bar p$ (Assumption~\ref{ass:BC}) and $a\ge b$. Iterate for $t\le T$.
\end{proof}

\begin{lemma}[Contrast recursion and monotone decrease to a fixed point]\label{lem:monotone-benchmark-contrast-sec5}
Assume Assumption~\ref{ass:BC}. Let $s_t:=w_{\mathrm{out}}^{(t)}/w_{\mathrm{in}}^{(t)}\in(0,1]$. Then
$s_0=1$ and
\[
s_{t+1}=f_n(s_t),
\qquad
f_n(s):=\frac{2(n-1)p_0p_1\,s}{(n-2)p_0^2+n p_1^2\,s}.
\]
Moreover, $f_n$ is increasing on $(0,1]$, and $f_n(s)/s$ is strictly decreasing on $(0,1]$.
Define
\[
s_\star:=\max\Big\{0,\ \frac{2(n-1)p_0p_1-(n-2)p_0^2}{n p_1^2}\Big\}\in[0,1),
\qquad
\Delta_t^{(n)} := \frac{(n-1)p_0 - n p_1 s_t}{(n-1)p_0 + n p_1 s_t}.
\]
Then $(s_t)$ is nonincreasing and converges to $s_\star$, and $t\mapsto \Delta_t^{(n)}$ is nondecreasing.
\end{lemma}

\begin{proof}
Differentiate to get $f_n'(s)>0$. Also
\[
\frac{f_n(s)}{s}=\frac{2(n-1)p_0p_1}{(n-2)p_0^2+n p_1^2 s}
\]
is strictly decreasing in $s$. The inequality $f_n(s)\le s$ is equivalent to $s\ge s_\star$.
Since $s_0=1\ge s_\star$ and $f_n([s_\star,1])\subseteq [s_\star,1]$, we get $s_{t+1}\le s_t$ and $s_t\downarrow s_\star$.
Finally, $\Delta_t^{(n)}$ is strictly decreasing in $s_t$.
\end{proof}

\subsection{Population benchmark eigengap monotonicity}\label{subsec:benchmark-gap-sec5}

Let $P^{\pop}:=\E[A\mid\sigma]$. Define the population benchmark
\[
W^{\star,(t)}_{\pop}:=K^{(t)}\circ P^{\pop},
\qquad
L^{\star,(t)}_{\pop}:=L(W^{\star,(t)}_{\pop}),
\qquad
\Gamma_{\pop}^{(t)}:=\lambda_3(L^{\star,(t)}_{\pop})-\lambda_2(L^{\star,(t)}_{\pop}).
\]

\begin{lemma}[Population benchmark gap equals $\Delta_t^{(n)}+O(1/n)$]\label{lem:pop-gap-monotone-sec5}
Fix $T\in\mathbb N$ and assume Assumption~\ref{ass:BC}. For each $t\le T$,
\[
\Gamma_{\pop}^{(t)}=\Delta_t^{(n)}+O\!\Big(\frac{1}{n}\Big),
\]
hence $t\mapsto\Gamma_{\pop}^{(t)}$ is nondecreasing up to $O(1/n)$.
\end{lemma}

\begin{proof}
$W^{\star,(t)}_{\pop}$ has the $2\times2$ block form of Lemma~\ref{lem:block-spectrum-fixed} with
$a_t:=p_0 w_{\mathrm{in}}^{(t)}$ and $b_t:=p_1 w_{\mathrm{out}}^{(t)}$. Thus
\[
\lambda_2(L^{\star,(t)}_{\pop})=1-\frac{(n-1)a_t-nb_t}{(n-1)a_t+nb_t}=1-\Delta_t^{(n)},
\qquad
\lambda_3(L^{\star,(t)}_{\pop})=1+\frac{a_t}{(n-1)a_t+nb_t}=1+O\!\Big(\frac1n\Big).
\]
Subtract.
\end{proof}

\subsection{A uniform good event on the base graph}\label{subsec:good-event-sec5}

Write $\Gamma(x):=\{z:A_{xz}=1\}$ for the (unweighted) neighborhood. For an edge $x\sim y$, define
\[
U_x:=\Gamma(x)\setminus(\Gamma(y)\cup\{y\}),\qquad
U_y:=\Gamma(y)\setminus(\Gamma(x)\cup\{x\}).
\]
Define the exclusive neighborhoods by \emph{weight type} (relative to the endpoint):
\[
U_{\rm in}(x,y):=U_x\cap B_{\sigma(x)},\qquad
U_{\rm out}(x,y):=U_x\cap B_{\sigma(x)}^{c},
\]
\[
V_{\rm in}(x,y):=U_y\cap B_{\sigma(y)},\qquad
V_{\rm out}(x,y):=U_y\cap B_{\sigma(y)}^{c}.
\]
Define imbalances
\[
m_{\rm in}(x,y):=\big||U_{\rm in}(x,y)|-|V_{\rm in}(x,y)|\big|,\qquad
m_{\rm out}(x,y):=\big||U_{\rm out}(x,y)|-|V_{\rm out}(x,y)|\big|,\qquad
m(x,y):=m_{\rm in}(x,y)+m_{\rm out}(x,y).
\]

\begin{lemma}[Uniform good event $\mathcal E_T$]\label{lem:ET-sec5}
Fix $T\in\mathbb N$ and assume Assumptions~\ref{ass:BC} and~\ref{ass:MDT}. There exists an event $\mathcal E_T$ with
$\Pr(\mathcal E_T)\ge 1-C_{T,\rho}n^{-2}$ such that on $\mathcal E_T$:
\begin{enumerate}
\item[(i)] Lemmas~\ref{lem:blockwise-deg} and~\ref{lem:blockwise-codeg} hold uniformly over all vertices and pairs.
\item[(ii)] For every edge $x\sim y$, the bipartite graphs induced by
$U_{\rm in}(x,y)\cup V_{\rm in}(x,y)$ and by $U_{\rm out}(x,y)\cup V_{\rm out}(x,y)$ each contain a matching
saturating the smaller side.
\item[(iii)] Uniformly over all edges $x\sim y$, $m(x,y)\lesssim \sqrt{n\bar p\log n}$.
\end{enumerate}
\end{lemma}

\begin{proof}
(i) is Lemmas~\ref{lem:blockwise-deg} and~\ref{lem:blockwise-codeg} plus a union bound; Assumption~\ref{ass:MDT} implies
$n\bar p\gg\log n$ and $n\bar p^2\gg\log n$ (Remark~\ref{rem:MDT-conseq}).

(ii) Fix an ordered edge $(x,y)$ and consider the bipartite graph between $U_{\rm in}(x,y)$ and $V_{\rm in}(x,y)$.
Conditional on $\sigma,\Gamma(x),\Gamma(y)$, edges between these sets are independent and each has success probability
at least $p_1\asymp_\rho \bar p$. On the event in (i), uniformly over $x\sim y$ we have
$|U_{\rm in}(x,y)|,|V_{\rm in}(x,y)|\asymp_\rho n\bar p$, and the aspect ratio is bounded by a constant $\Lambda(\rho)$.
Also $p_1|V_{\rm in}|\asymp_\rho n\bar p^2\gg \log n$. Lemma~\ref{lem:bip-matching} applies (take $q=10$).
The same argument applies to $(U_{\rm out},V_{\rm out})$. A union bound over $O(n^2)$ ordered edges yields
overall failure probability at most $C\,n^2\,(c\,n\bar p)^{-10}=o(1)$, and in particular $\le n^{-2}$ for all
sufficiently large $n$.

(iii) Each of $|U_{\rm in}|,|V_{\rm in}|,|U_{\rm out}|,|V_{\rm out}|$ is a signed linear combination of blockwise
degrees and blockwise co-degrees, hence concentrates to $O(\sqrt{n\bar p\log n})$ uniformly on the event in (i).
Summing gives $m(x,y)\lesssim \sqrt{n\bar p\log n}$ uniformly.
\end{proof}

\subsection{Spectral concentration for two-level reweighting}\label{subsec:twolevel-spec-sec5}

\begin{lemma}[Spectral concentration for two-level reweighted SBM]\label{lem:conc-twolevel-sec5}
Fix $T\in\mathbb N$ and assume Assumptions~\ref{ass:BC} and~\ref{ass:MDT}.
Let $(a,b)$ satisfy $0<b\le a$ and $a\asymp_{T,\rho}\bar p$, $b\asymp_{T,\rho}\bar p$.
Set
\[
W^\star(a,b)=K(a,b)\circ A,
\qquad
W^\star_{\pop}(a,b)=K(a,b)\circ P^{\pop},
\qquad
L(W):=I-\Deg(W)^{-1/2}W\Deg(W)^{-1/2}.
\]
Then there exists $C_{T,\rho}<\infty$ such that with probability at least $1-C_{T,\rho}n^{-2}$,
\[
\bigl\|L(W^\star(a,b)) - L(W^\star_{\pop}(a,b))\bigr\|_{\op}\ \le\ C_{T,\rho}\,\varepsilon_n.
\]
Consequently, for each fixed $k$,
\[
\bigl|\lambda_k(L(W^\star(a,b))) - \lambda_k(L(W^\star_{\pop}(a,b)))\bigr|
\ \le\ C_{T,\rho}\,\varepsilon_n .
\]
\end{lemma}

\begin{proof}
Condition on $\sigma$ and set $\Delta W:=K(a,b)\circ(A-P^{\pop})$.
Matrix Bernstein gives $\|\Delta W\|_{\op}\lesssim \sqrt{n\bar p^3\log n}$ w.h.p. since $\|K(a,b)\|_{\max}\asymp \bar p$.
Also, $d_x(W^\star(a,b))\asymp n\bar p^2$ uniformly and
$\|d(W^\star(a,b))-d(W^\star_{\pop}(a,b))\|_\infty\lesssim \sqrt{n\bar p^3\log n}$ w.h.p.
Apply Lemma~\ref{lem:normlap-perturb} with $m\asymp n\bar p^2$ and note
$\|W^\star(a,b)\|_{\op},\|W^\star_{\pop}(a,b)\|_{\op}\lesssim n\bar p^2$.
This yields
\[
\|L(W^\star(a,b))-L(W^\star_{\pop}(a,b))\|_{\op}
\ \lesssim\ \frac{\sqrt{n\bar p^3\log n}}{n\bar p^2}
=\varepsilon_n.
\]
Eigenvalue control follows from Weyl.
\end{proof}

\subsection{Curvature of two-level benchmarks and concentration}\label{subsec:overlap-sec5}

For a weighted adjacency $W$ supported on $A$, write $d_x(W):=\sum_{z\sim x}W_{xz}$ and
$p_{x,W}(z):=W_{xz}/d_x(W)$ for $z\sim x$.
For an edge $x\sim y$ define the overlap functional
\begin{equation}\label{eq:overlap-sec5}
\operatorname{Ov}_{xy}(W)
:= p_{x,W}(y) + p_{y,W}(x) + \sum_{z\in \Gamma(x)\cap \Gamma(y)} \min\{p_{x,W}(z),p_{y,W}(z)\},
\end{equation}
and set
\[
p_{\max}(W;x,y) := \max\Bigl\{\max_{u\sim x}p_{x,W}(u),\ \max_{v\sim y}p_{y,W}(v)\Bigr\}.
\]
Also define the local degree mismatch
\[
\Delta_{\deg}(W;x,y):=\frac{|d_x(W)-d_y(W)|}{\min\{d_x(W),d_y(W)\}}.
\]

\begin{lemma}[Deterministic overlap upper bound]\label{lem:overlap-upper-sec5}
For any weighted $W$ supported on $A$ and any edge $x\sim y$,
\[
\kappa_W(x,y)\ \le\ \operatorname{Ov}_{xy}(W).
\]
\end{lemma}

\begin{proof}
Fix $\alpha\in[1/2,1)$ and write $m^\alpha_{x,W}:=\alpha\delta_x+(1-\alpha)\sum_{z\sim x}p_{x,W}(z)\delta_z$.
For any coupling $\pi$ of $(m^\alpha_{x,W},m^\alpha_{y,W})$, since $d(u,v)\ge 1$ for $u\neq v$,
\[
\sum_{u,v}\pi(u,v)d(u,v)\ \ge\ 1-\sum_u\pi(u,u).
\]
Maximizing diagonal mass over couplings yields $\sum_u \min\{m^\alpha_{x,W}(u),m^\alpha_{y,W}(u)\}$.
For $\alpha\ge 1/2$ we compute this overlap explicitly. Write $\beta:=1-\alpha$.
At $x$, we have $m^\alpha_{x,W}(x)=\alpha$ and $m^\alpha_{y,W}(x)=\beta\,p_{y,W}(x)\le\beta\le\alpha$, hence
$\min\{m^\alpha_{x,W}(x),m^\alpha_{y,W}(x)\}=\beta\,p_{y,W}(x)$. Similarly,
$\min\{m^\alpha_{x,W}(y),m^\alpha_{y,W}(y)\}=\beta\,p_{x,W}(y)$.
For $z\in\Gamma(x)\cap\Gamma(y)$, $m^\alpha_{x,W}(z)=\beta p_{x,W}(z)$ and $m^\alpha_{y,W}(z)=\beta p_{y,W}(z)$, so
the contribution is $\beta\min\{p_{x,W}(z),p_{y,W}(z)\}$, and all other vertices contribute $0$.
Summing yields $\sum_u \min\{m^\alpha_{x,W}(u),m^\alpha_{y,W}(u)\}=\beta\,\operatorname{Ov}_{xy}(W)$.
Thus $\kappa_{\alpha,W}(x,y)=1-W_1(m^\alpha_{x,W},m^\alpha_{y,W})\le (1-\alpha)\operatorname{Ov}_{xy}(W)$.
Divide by $(1-\alpha)$ and let $\alpha\uparrow 1$.
\end{proof}

\begin{lemma}[Overlap lower bound from type matchings (two-level case)]\label{lem:overlap-lower-from-matchings-sec5}
Assume $W$ is of the \emph{two-level} form $W=W^\star(a,b)=K(a,b)\circ A$.
Let $x\sim y$ and let $U_{\rm in},U_{\rm out},V_{\rm in},V_{\rm out}$ be as in Section~\ref{subsec:good-event-sec5},
with $m(x,y)$ defined there. Assume that the bipartite graphs induced by $U_{\rm in}\cup V_{\rm in}$ and
$U_{\rm out}\cup V_{\rm out}$ contain matchings saturating the smaller side.
Then for every $\alpha\in[1/2,1)$,
\[
\kappa_{\alpha,W}(x,y)
\ \ge\ (1-\alpha)\operatorname{Ov}_{xy}(W)\ -\ 3(1-\alpha)\Big(p_{\max}(W;x,y)\,m(x,y)\ +\ 2\,\Delta_{\deg}(W;x,y)\Big),
\]
and hence, letting $\alpha\uparrow 1$,
\[
\kappa_W(x,y)\ \ge\ \operatorname{Ov}_{xy}(W)\ -\ 3\,p_{\max}(W;x,y)\,m(x,y)\ -\ 6\,\Delta_{\deg}(W;x,y).
\]
\end{lemma}

\begin{proof}
Fix $\alpha\in[1/2,1)$ and write $\beta:=1-\alpha$. We construct a coupling of
$m^\alpha_{x,W}:=\alpha\delta_x+\beta\sum_{z\sim x}p_{x,W}(z)\delta_z$ and
$m^\alpha_{y,W}:=\alpha\delta_y+\beta\sum_{z\sim y}p_{y,W}(z)\delta_z$.

\smallskip\noindent
\emph{Step 1 (diagonal mass equals overlap).}
Set $\pi(z,z):=\min\{m^\alpha_{x,W}(z),m^\alpha_{y,W}(z)\}$ for all $z$.
For $\alpha\ge 1/2$, the minima at $x$ and $y$ equal $\beta p_{y,W}(x)$ and $\beta p_{x,W}(y)$, and at common neighbors
equal $\beta\min\{p_{x,W}(z),p_{y,W}(z)\}$. Thus the total diagonal mass equals $\beta\,\operatorname{Ov}_{xy}(W)$.

\smallskip\noindent
\emph{Step 2 (match exclusive types; bound total leftover mass).}
Because $W$ is two-level, neighbor probabilities are constant on each type set:
on $U_{\rm in}$, $p_{x,W}(\cdot)=a/d_x(W)$; on $U_{\rm out}$, $p_{x,W}(\cdot)=b/d_x(W)$; and similarly for $y$.

Consider $(U_{\rm in},V_{\rm in})$.
Let $\phi_{\rm in}$ be an injection from the smaller side into the larger along matched edges.
For each matched pair $(u,\phi_{\rm in}(u))$, couple an amount
$\beta\,\min\{a/d_x(W),a/d_y(W)\}$ from the source vertex to its partner (cost $1$ since matched pairs are edges).
Unmatched vertices on the larger side contribute leftover mass at most $\beta\,p_{\max}(W;x,y)\,m_{\rm in}(x,y)$.

Additionally, even among matched vertices, a normalization mismatch $d_x(W)\neq d_y(W)$ causes leftover mass.
Its \emph{total} contribution is bounded by
\[
\beta\cdot \min\{|U_{\rm in}|,|V_{\rm in}|\}\cdot \Big|\frac{a}{d_x(W)}-\frac{a}{d_y(W)}\Big|
\ \le\ \beta\,\frac{|d_x(W)-d_y(W)|}{\min\{d_x(W),d_y(W)\}}
=\beta\,\Delta_{\deg}(W;x,y),
\]
since $\min\{|U_{\rm in}|,|V_{\rm in}|\}\le \min\{d_x(W),d_y(W)\}/a$ (because $a|U_{\rm in}|\le d_x(W)$ and $a|V_{\rm in}|\le d_y(W)$).
Perform the same construction on $(U_{\rm out},V_{\rm out})$ to obtain leftover mass at most
$\beta\,p_{\max}(W;x,y)m_{\rm out}(x,y)+\beta\,\Delta_{\deg}(W;x,y)$.

Therefore the \emph{total} leftover mass after Step 2 is at most
\[
\ell\ \le\ \beta\,p_{\max}(W;x,y)\,m(x,y)\ +\ 2\beta\,\Delta_{\deg}(W;x,y).
\]

\smallskip\noindent
\emph{Step 3 (route leftover at cost $\le 3$).}
All leftover mass is supported on $\Gamma(x)\cup\Gamma(y)$. Given any $s\in\Gamma(x)$ and $t\in\Gamma(y)$,
$d(s,t)\le 3$ via $s-x-y-t$. Hence we can complete the coupling by routing all leftover mass at per-unit cost at most $3$.

\smallskip\noindent
Step 4 (conclude). The mass placed on the diagonal equals $\beta\,\operatorname{Ov}_{xy}(W)$.
Among the remaining mass, all transport performed along matched edges costs $1$ per unit, while the leftover mass $\ell$
is routed at per-unit cost at most $3$. Therefore
\[
W_1(m^\alpha_{x,W},m^\alpha_{y,W})
\le (1-\beta\,\operatorname{Ov}_{xy}(W)-\ell)\cdot 1 + 3\ell
= 1-\beta\,\operatorname{Ov}_{xy}(W)+2\ell
\le 1-\beta\,\operatorname{Ov}_{xy}(W)+3\ell,
\]
Thus
\[
\kappa_{\alpha,W}(x,y)=1-W_1(m^\alpha_{x,W},m^\alpha_{y,W})
\ \ge\ \beta\,\operatorname{Ov}_{xy}(W)-3\ell,
\]
and substituting the bound on $\ell$ yields the $\alpha$-lazy inequality. Divide by $\beta$ and let $\alpha\uparrow 1$.
\end{proof}

\begin{lemma}[Kernel maximum for two-level benchmarks]\label{lem:pmax-benchmark-sec5}
Fix $T\in\mathbb N$ and assume $c_{T,\rho}\bar p\le b\le a\le C_{T,\rho}\bar p$.
On $\mathcal E_T$, for every edge $x\sim y$,
\[
p_{\max}\!\bigl(W^\star(a,b);x,y\bigr)\ \lesssim_{T,\rho}\ \frac{1}{n\bar p}.
\]
\end{lemma}

\begin{proof}
On $\mathcal E_T$(i), blockwise degrees are $\Theta(n\bar p)$ uniformly, hence
$d_x(W^\star(a,b)) = a\,|\Gamma(x)\cap B_{\sigma(x)}| + b\,|\Gamma(x)\cap B_{\sigma(x)}^c|
\asymp_{T,\rho} n\bar p^2$.
Each nonzero weight is $\le a\lesssim \bar p$, so each neighbor probability is
$\lesssim \bar p/(n\bar p^2)=1/(n\bar p)$, uniformly.
\end{proof}

\begin{lemma}[Two-scale ratio expansion]\label{lem:ratio-expansion-np2-sec5}
Let $\bar p=\bar p(n)\in(0,1]$. Let $A=\mu_A+\delta_A$ and $B=\mu_B+\delta_B$ be real random variables.
Assume that for some constants $c,C_1,C_2,C_3>0$, on an event $\mathcal G$:
\begin{enumerate}
\item[(I)] (\emph{$n\bar p^2$--scale}) $\mu_B \ge c\,n\bar p^{\,2}$, $|\mu_A|\le C_1\,n\bar p^{\,3}$, and
$|\delta_B|\le C_2\sqrt{n\bar p^{\,3}\log n}$, $|\delta_A|\le C_3\sqrt{n\bar p^{\,4}\log n}$; or
\item[(II)] (\emph{$n\bar p$--scale}) $\mu_B \ge c\,n\bar p$, $|\mu_A|\le C_1\,n\bar p^{\,2}$, and
$|\delta_B|\le C_2\sqrt{n\bar p\log n}$, $|\delta_A|\le C_3\sqrt{n\bar p^{\,2}\log n}$.
\end{enumerate}
Then on $\mathcal G$, for all sufficiently large $n$,
\[
\left|\frac{A}{B}-\frac{\mu_A}{\mu_B}\right|\ \le\ C\,\varepsilon_n,
\qquad
\varepsilon_n=\sqrt{\frac{\log n}{n\bar p}},
\]
where $C$ depends only on $(c,C_1,C_2,C_3)$.
\end{lemma}

\begin{proof}
Write
\[
\frac{A}{B}-\frac{\mu_A}{\mu_B}
=\frac{\delta_A\mu_B-\mu_A\delta_B}{B\mu_B}.
\]
On $\mathcal G$, in either regime (I) or (II), $|\delta_B|=o(\mu_B)$ under Assumption~\ref{ass:MDT},
so $|B|\ge \mu_B/2$ for large $n$. Hence
\[
\left|\frac{A}{B}-\frac{\mu_A}{\mu_B}\right|
\le \frac{2}{\mu_B^2}\Big(|\delta_A|\mu_B+|\mu_A||\delta_B|\Big)
=2\Big(\frac{|\delta_A|}{\mu_B}+\frac{|\mu_A||\delta_B|}{\mu_B^2}\Big).
\]
In case (I),
\[
\frac{|\delta_A|}{\mu_B}\ \lesssim\ \frac{\sqrt{n\bar p^4\log n}}{n\bar p^2}
=\bar p\sqrt{\frac{\log n}{n}}\ \le\ \varepsilon_n,
\]
and
\[
\frac{|\mu_A||\delta_B|}{\mu_B^2}
\ \lesssim\ \frac{n\bar p^3\cdot \sqrt{n\bar p^3\log n}}{n^2\bar p^4}
=\sqrt{\frac{\bar p\log n}{n}}\ \le\ \varepsilon_n.
\]
Case (II) is identical after replacing $\bar p^2$ by $\bar p$ in the scale bookkeeping.
\end{proof}

\begin{lemma}[Overlap concentration for benchmark iterates]\label{lem:ov-conc-sec5}
Fix $T\in\mathbb N$ and assume Assumptions~\ref{ass:BC} and~\ref{ass:MDT}. On $\mathcal E_T$, for all
$t\le T-1$ and all edges $x\sim y$,
\[
\Big|\operatorname{Ov}_{xy}(W^{\star,(t)})-K^{(t+1)}_{xy}\Big|\ \le\ C_{T,\rho}\,\varepsilon_n.
\]
\end{lemma}

\begin{proof}
Fix $t\le T-1$ and edge $x\sim y$. Write $a_t:=w_{\mathrm{in}}^{(t)}$, $b_t:=w_{\mathrm{out}}^{(t)}$ and
$W^\star:=W^\star(a_t,b_t)$.

On $\mathcal E_T$(i), for $t\ge 1$ we have $a_t,b_t\asymp \bar p$ (Lemma~\ref{lem:uniform-order-benchmark-sec5}) and
\[
d_x(W^\star)=(n-1)p_0a_t+np_1b_t\ \pm\ O\!\big(\sqrt{n\bar p^3\log n}\big),
\]
uniformly in $x$, hence $d_x(W^\star)=\Theta(n\bar p^2)$ and $d_x(W^\star)/d_y(W^\star)=1+O(\varepsilon_n)$ uniformly.

\smallskip\noindent
\emph{Within-block edge.}
If $\sigma(x)=\sigma(y)$, then on common neighbors in the home block both kernels contribute $a_t/d_x$ and $a_t/d_y$,
and on the opposite block both contribute $b_t/d_x$ and $b_t/d_y$, so the min uses the smaller denominator.
Thus
\[
\operatorname{Ov}_{xy}(W^\star)
=
\frac{a_t\,D_{xy}^{(\text{home})}+b_t\,D_{xy}^{(\text{opp})}}{\max\{d_x(W^\star),d_y(W^\star)\}}
\ +\ O\!\Big(\frac{1}{n\bar p}\Big),
\]
where the endpoint term $p_{x,W}(y)+p_{y,W}(x)=O((n\bar p)^{-1})=o(\varepsilon_n)$ since $n\bar p\gg\log n$
under Assumption~\ref{ass:MDT}.
Apply Lemma~\ref{lem:ratio-expansion-np2-sec5} (regime (I) for $t\ge 1$, regime (II) for $t=0$) together with
blockwise co-degree concentration on $\mathcal E_T$(i). This yields
\[
\operatorname{Ov}_{xy}(W^{\star,(t)})
=
\varphi_{\mathrm{in}}^{(n)}(a_t,b_t)\ \pm\ C_{T,\rho}\varepsilon_n
=
K^{(t+1)}_{xy}\ \pm\ C_{T,\rho}\varepsilon_n.
\]

\smallskip\noindent
\emph{Cross-block edge.}
If $\sigma(x)\neq\sigma(y)$, then on common neighbors in $B_{\sigma(x)}$ the kernels contribute
$a_t/d_x$ versus $b_t/d_y$, and in $B_{\sigma(y)}$ they contribute $b_t/d_x$ versus $a_t/d_y$.
For $t\ge 1$, Lemma~\ref{lem:monotone-benchmark-contrast-sec5} implies $b_t/a_t\le 1-c(\rho)$, while
$d_x/d_y=1+O(\varepsilon_n)$ uniformly, so for all large $n$ we have
$a_t/d_x \ge b_t/d_y$ and $a_t/d_y \ge b_t/d_x$; hence the minima select $b_t/d_y$ on $B_{\sigma(x)}$
and $b_t/d_x$ on $B_{\sigma(y)}$.
Therefore $\operatorname{Ov}_{xy}(W^\star)$ is a sum of two ratios of the form
$b_t D_{xy}^{(B)}/d_{\cdot}(W^\star)$ plus endpoint terms $o(\varepsilon_n)$, and another application of
Lemma~\ref{lem:ratio-expansion-np2-sec5} yields
\[
\operatorname{Ov}_{xy}(W^{\star,(t)})
=
\varphi_{\mathrm{out}}^{(n)}(a_t,b_t)\ \pm\ C_{T,\rho}\varepsilon_n
=
K^{(t+1)}_{xy}\ \pm\ C_{T,\rho}\varepsilon_n.
\]
\end{proof}

\begin{proposition}[Two-level benchmark curvature concentrates around $\Phi_n$]\label{prop:benchmark-curv-conc-sec5}
Fix $T\in\mathbb N$ and assume Assumptions~\ref{ass:BC} and~\ref{ass:MDT}. On $\mathcal E_T$, for all
$t\le T-1$ and all edges $x\sim y$,
\[
\bigl|\kappa_{W^{\star,(t)}}(x,y) - K^{(t+1)}_{xy}\bigr|
\ \le\ C_{T,\rho}\,\varepsilon_n,
\qquad\text{equivalently}\qquad
\bigl\|\kappa_{W^{\star,(t)}}\circ A - W^{\star,(t+1)}\bigr\|_{\max}
\ \le\ C_{T,\rho}\,\varepsilon_n.
\]
\end{proposition}

\begin{proof}
Fix $t\le T-1$ and an edge $x\sim y$.

\emph{Upper bound.} Lemma~\ref{lem:overlap-upper-sec5} gives
$\kappa_{W^{\star,(t)}}(x,y)\le \operatorname{Ov}_{xy}(W^{\star,(t)})$.

\emph{Lower bound.} On $\mathcal E_T$, Lemma~\ref{lem:ET-sec5}(ii) provides the type matchings, hence
Lemma~\ref{lem:overlap-lower-from-matchings-sec5} applies to $W^{\star,(t)}$:
\[
\kappa_{W^{\star,(t)}}(x,y)
\ge \operatorname{Ov}_{xy}(W^{\star,(t)})
- 3\,p_{\max}(W^{\star,(t)};x,y)\,m(x,y)
- 6\,\Delta_{\deg}(W^{\star,(t)};x,y).
\]
By Lemma~\ref{lem:pmax-benchmark-sec5} and Lemma~\ref{lem:ET-sec5}(iii),
$p_{\max}(W^{\star,(t)};x,y)\,m(x,y)\lesssim \varepsilon_n$ uniformly.
Also on $\mathcal E_T$(i), $d_x(W^{\star,(t)})=\Theta(n\bar p^2)$ and
$|d_x(W^{\star,(t)})-d_y(W^{\star,(t)})|\lesssim \sqrt{n\bar p^3\log n}$ uniformly, hence
$\Delta_{\deg}(W^{\star,(t)};x,y)\lesssim \varepsilon_n$ uniformly.
Therefore
\[
\kappa_{W^{\star,(t)}}(x,y)=\operatorname{Ov}_{xy}(W^{\star,(t)})\pm C_{T,\rho}\varepsilon_n.
\]
Combine with Lemma~\ref{lem:ov-conc-sec5}.
\end{proof}

\subsection{Iterate tracking in max norm}\label{subsec:tracking-sec5}

\begin{lemma}[Curvature stability via neighbor TV]\label{lem:curv-stability-sec5}
Fix $T\in\mathbb N$ and work on $\mathcal E_T$.
Let $t\in\{1,\dots,T\}$ and suppose $W$ is supported on $A$ and satisfies
\[
\|W-W^{\star,(t)}\|_{\max}\le \delta,
\qquad
\min_x d_x(W)\ \ge\ c_{T,\rho}\,n\bar p^2.
\]
Then for every edge $x\sim y$,
\[
\bigl|\kappa_W(x,y)-\kappa_{W^{\star,(t)}}(x,y)\bigr|
\ \le\ C_{T,\rho}\,\frac{\delta}{\bar p}.
\]
\end{lemma}

\begin{proof}
Fix $x\sim y$ and $\alpha\in[1/2,1)$ and write $\beta:=1-\alpha$.
By the triangle inequality for $W_1$,
\[
\Big|W_1(m^\alpha_{x,W},m^\alpha_{y,W})-W_1(m^\alpha_{x,W^\star},m^\alpha_{y,W^\star})\Big|
\le W_1(m^\alpha_{x,W},m^\alpha_{x,W^\star})+W_1(m^\alpha_{y,W},m^\alpha_{y,W^\star}).
\]
Hence
\[
|\kappa_{\alpha,W}(x,y)-\kappa_{\alpha,W^\star}(x,y)|
\le W_1(m^\alpha_{x,W},m^\alpha_{x,W^\star})+W_1(m^\alpha_{y,W},m^\alpha_{y,W^\star}).
\]
The supports of $m^\alpha_{x,W}$ and $m^\alpha_{x,W^\star}$ are contained in $\{x\}\cup\Gamma(x)$, whose
graph diameter is at most $2$, so
\[
W_1(m^\alpha_{x,W},m^\alpha_{x,W^\star})
\le 2\,\|m^\alpha_{x,W}-m^\alpha_{x,W^\star}\|_{\mathrm{TV}}
=2\beta\,\|p_{x,W}-p_{x,W^\star}\|_{\mathrm{TV}},
\]
and similarly for $y$. Dividing by $\beta$ and letting $\alpha\uparrow 1$ yields
\[
|\kappa_W(x,y)-\kappa_{W^\star}(x,y)|
\le 2\Big(\|p_{x,W}-p_{x,W^\star}\|_{\mathrm{TV}}+\|p_{y,W}-p_{y,W^\star}\|_{\mathrm{TV}}\Big).
\]

To bound $\|p_{x,W}-p_{x,W^\star}\|_{\mathrm{TV}}$, write $d_x:=d_x(W)$ and $d_x^\star:=d_x(W^\star)$.
Then
\[
\sum_{u\sim x}\left|\frac{W_{xu}}{d_x}-\frac{W^\star_{xu}}{d_x^\star}\right|
\le \frac{1}{d_x}\sum_{u\sim x}|W_{xu}-W^\star_{xu}|
     +\left|\frac{1}{d_x}-\frac{1}{d_x^\star}\right|\sum_{u\sim x}W^\star_{xu}.
\]
Since $\sum_{u\sim x}W^\star_{xu}=d_x^\star$ and $|d_x-d_x^\star|\le \sum_{u\sim x}|W_{xu}-W^\star_{xu}|\le D_x\delta$,
we get
\[
\sum_{u\sim x}\left|\frac{W_{xu}}{d_x}-\frac{W^\star_{xu}}{d_x^\star}\right|
\le \frac{D_x\delta}{d_x}+\frac{|d_x-d_x^\star|}{d_x}
\le \frac{2D_x\delta}{d_x}.
\]
Therefore $\|p_{x,W}-p_{x,W^\star}\|_{\mathrm{TV}}\le D_x\delta/d_x$.
On $\mathcal E_T$, $D_x\lesssim n\bar p$ uniformly, and by assumption $d_x\ge c\,n\bar p^2$, hence
$\|p_{x,W}-p_{x,W^\star}\|_{\mathrm{TV}}\lesssim \delta/\bar p$ uniformly; similarly for $y$.
\end{proof}

\begin{theorem}[Iterates track the observed two-level benchmark]\label{thm:iterate-tracking-sec5}
Fix $T\in\mathbb N$ and assume Assumptions~\ref{ass:BC},~\ref{ass:MDT}, and \eqref{eq:MDT-T-sec5}.
There exists $C_{T,\rho}<\infty$ such that, with probability at least $1-C_{T,\rho}n^{-2}$, simultaneously for all
$t\le T$,
\[
\bigl\|W^{(t)}-W^{\star,(t)}\bigr\|_{\max}\ \le\ C_{T,\rho}\,\frac{\varepsilon_n}{\bar p^{\,T-1}}.
\]
More sharply, for each $t\in\{1,\dots,T\}$,
\[
\bigl\|W^{(t)}-W^{\star,(t)}\bigr\|_{\max}\ \le\ C_{t,\rho}\,\frac{\varepsilon_n}{\bar p^{\,t-1}}.
\]
\end{theorem}

\begin{proof}
Work on $\mathcal E_T$. Set $\delta_t:=\|W^{(t)}-W^{\star,(t)}\|_{\max}$ so $\delta_0=0$.
Since $W^{(0)}=W^{\star,(0)}=A$, Proposition~\ref{prop:benchmark-curv-conc-sec5} at $t=0$ gives
$\delta_1\le C\varepsilon_n$.

Fix $t\in\{1,\dots,T-1\}$. For any edge $x\sim y$,
\[
|W^{(t+1)}_{xy}-W^{\star,(t+1)}_{xy}|
=
|\kappa_{W^{(t)}}(x,y)-K^{(t+1)}_{xy}|
\le
|\kappa_{W^{(t)}}(x,y)-\kappa_{W^{\star,(t)}}(x,y)|
+
|\kappa_{W^{\star,(t)}}(x,y)-K^{(t+1)}_{xy}|.
\]
The second term is $\le C\varepsilon_n$ by Proposition~\ref{prop:benchmark-curv-conc-sec5}.
For the first term, we apply Lemma~\ref{lem:curv-stability-sec5}.
It remains to verify the degree lower bound needed there for $W^{(t)}$.

On $\mathcal E_T$ and for $t\ge 1$, $\min_x d_x(W^{\star,(t)})\asymp n\bar p^2$ (since $w_{\rm in}^{(t)},w_{\rm out}^{(t)}\asymp\bar p$
and degrees are $\Theta(n\bar p)$). Also,
\[
\min_x d_x(W^{(t)})
\ \ge\ \min_x d_x(W^{\star,(t)})-(\max_x D_x)\,\delta_t
\ \ge\ c\,n\bar p^2 - C(n\bar p)\delta_t.
\]
Under the induction bound $\delta_t\le C_{t,\rho}\varepsilon_n/\bar p^{t-1}$ and \eqref{eq:MDT-T-sec5},
we have $(n\bar p)\delta_t=o(n\bar p^2)$ uniformly for $t\le T$, hence for all large $n$,
$\min_x d_x(W^{(t)})\ge c' n\bar p^2$ and Lemma~\ref{lem:curv-stability-sec5} applies.

Therefore, uniformly over edges,
\[
|\kappa_{W^{(t)}}(x,y)-\kappa_{W^{\star,(t)}}(x,y)|\le C\,\frac{\delta_t}{\bar p}.
\]
Taking maxima over edges gives the recursion
\[
\delta_{t+1}\ \le\ C\frac{\delta_t}{\bar p} + C\varepsilon_n.
\]
Iterating this recursion for fixed $T$ yields $\delta_t\le C_{t,\rho}\varepsilon_n/\bar p^{t-1}$, and hence
$\sup_{t\le T}\delta_t\le C_{T,\rho}\varepsilon_n/\bar p^{T-1}$.
\end{proof}

\begin{corollary}[Uniform positivity of iterated weights and degrees]\label{cor:iter-positive-sec5}
Fix $T\in\mathbb N$ and assume Assumptions~\ref{ass:BC},~\ref{ass:MDT}, and \eqref{eq:MDT-T-sec5}.
Then with probability at least $1-C_{T,\rho}n^{-2}$, simultaneously for all $t\le T$ and all edges $x\sim y$,
\[
W^{(t)}_{xy}\ \ge\ c_{T,\rho}\,\bar p,
\]
and moreover $\min_x d_x(W^{(t)})\ge c_{T,\rho}\,n\bar p^2$ for all $t\ge 1$.
\end{corollary}

\begin{proof}
On $\mathcal E_T$, Lemma~\ref{lem:uniform-order-benchmark-sec5} gives
$\min\{w_{\mathrm{in}}^{(t)},w_{\mathrm{out}}^{(t)}\}\ge c_{T,\rho}\bar p$ for $t\ge 1$.
By Theorem~\ref{thm:iterate-tracking-sec5},
$\|W^{(t)}-W^{\star,(t)}\|_{\max}\le C_{T,\rho}\varepsilon_n/\bar p^{T-1}$.
Under \eqref{eq:MDT-T-sec5}, $\varepsilon_n/\bar p^{T-1}=o(\bar p)$, hence for all large $n$,
$C_{T,\rho}\varepsilon_n/\bar p^{T-1}\le \tfrac12 c_{T,\rho}\bar p$. Therefore
$W^{(t)}_{xy}\ge \tfrac12 c_{T,\rho}\bar p$ on every edge.

Summing over $\Theta(n\bar p)$ neighbors on $\mathcal E_T$ yields $\min_x d_x(W^{(t)})\gtrsim n\bar p^2$ for $t\ge 1$.
\end{proof}

\subsection{Normalized Laplacian tracking and eigengap stability}\label{subsec:lap-gap-sec5}

Recall $L(W):=I-\Deg(W)^{-1/2}W\Deg(W)^{-1/2}$.

\begin{corollary}[Laplacian tracking]\label{cor:laplacian-tracking-sec5}
Fix $T\in\mathbb N$ and assume Assumptions~\ref{ass:BC},~\ref{ass:MDT}, and \eqref{eq:MDT-T-sec5}.
With probability at least $1-C_{T,\rho}n^{-2}$, simultaneously for all $t\le T$,
\[
\bigl\|L(W^{(t)})-L(W^{\star,(t)})\bigr\|_{\op}\ \le\ C_{T,\rho}\,\eta_{n,T}.
\]
Consequently, for each fixed $k$,
\[
\bigl|\lambda_k(L(W^{(t)}))-\lambda_k(L(W^{\star,(t)}))\bigr|\ \le\ C_{T,\rho}\,\eta_{n,T}.
\]
\end{corollary}

\begin{proof}
Fix $t\in\{1,\dots,T\}$. By Corollary~\ref{cor:iter-positive-sec5}, on $\mathcal E_T$ we have
$\min_x d_x(W^{(t)})\asymp n\bar p^2$ and $\min_x d_x(W^{\star,(t)})\asymp n\bar p^2$ for large $n$.
Apply Lemma~\ref{lem:normlap-perturb} with $m\asymp n\bar p^2$.

Moreover, on $\mathcal E_T$(i),
\[
\|W^{(t)}-W^{\star,(t)}\|_{\op}
\le \|W^{(t)}-W^{\star,(t)}\|_{\infty}
\le (\max_x D_x)\,\|W^{(t)}-W^{\star,(t)}\|_{\max}
\ \lesssim\ (n\bar p)\cdot \frac{\varepsilon_n}{\bar p^{T-1}}.
\]
Dividing by $n\bar p^2$ yields the first perturbation term of order $\varepsilon_n/\bar p^T=\eta_{n,T}$.
The degree-difference term in Lemma~\ref{lem:normlap-perturb} is controlled similarly since
$\|d(W^{(t)})-d(W^{\star,(t)})\|_\infty\le (\max_x D_x)\|W^{(t)}-W^{\star,(t)}\|_{\max}$
and $\|W^{(t)}\|_{\op}+\|W^{\star,(t)}\|_{\op}\lesssim n\bar p^2$ for $t\ge 1$.
Eigenvalue control follows from Weyl.
\end{proof}

\begin{corollary}[Empirical eigengap is monotone up to vanishing errors]\label{cor:iterate-gap-monotone-sec5}
Fix $T\in\mathbb N$ and assume Assumptions~\ref{ass:BC},~\ref{ass:MDT}, and \eqref{eq:MDT-T-sec5}. Let
\[
\Gamma^{(t)} := \lambda_3(L(W^{(t)}))-\lambda_2(L(W^{(t)})).
\]
Then with probability at least $1-C_{T,\rho}n^{-2}$, for all $1\le t\le T$,
\[
\Gamma^{(t)} \ \ge\ \Gamma^{(t-1)}\ -\ C_{T,\rho}\eta_{n,T}\ -\ O\!\Big(\frac{1}{n}\Big).
\]
\end{corollary}

\begin{proof}
By Corollary~\ref{cor:laplacian-tracking-sec5}, uniformly for $t\le T$,
\[
\Gamma^{(t)}=\Gamma^{\star,(t)}+O(\eta_{n,T}),
\qquad
\Gamma^{\star,(t)}:=\lambda_3(L(W^{\star,(t)}))-\lambda_2(L(W^{\star,(t)})).
\]
For each fixed $t\ge 1$, Lemma~\ref{lem:conc-twolevel-sec5} with
$a=w_{\mathrm{in}}^{(t)}$, $b=w_{\mathrm{out}}^{(t)}$ gives $\Gamma^{\star,(t)}=\Gamma_{\pop}^{(t)}+O(\varepsilon_n)$.
Finally, Lemma~\ref{lem:pop-gap-monotone-sec5} gives $\Gamma_{\pop}^{(t)}=\Delta_t^{(n)}+O(1/n)$ and
$\Delta_t^{(n)}$ is nondecreasing by Lemma~\ref{lem:monotone-benchmark-contrast-sec5}.
Combine these relations for $t$ and $t-1$ and use $\varepsilon_n\le \eta_{n,T}$ (since $\bar p\le 1$ and $T\ge 1$).
\end{proof}

\begin{remark}[Strict monotonicity]\label{rem:strict-monotone-sec5}
For fixed $T$, Corollary~\ref{cor:iterate-gap-monotone-sec5} gives monotonicity up to $O(\eta_{n,T})+O(1/n)$.
Uniform \emph{strict} increments $\Gamma^{(t)}>\Gamma^{(t-1)}$ would require a deterministic lower bound on
$\Delta_t^{(n)}-\Delta_{t-1}^{(n)}$ that dominates $\eta_{n,T}$, which is not uniform under Assumption~\ref{ass:BC} alone.
\end{remark}

\section{Appendix}\label{sec:appendix-weighted} 
\addcontentsline{toc}{section}{Appendix}

We collect standard matrix perturbation tools used in Section~\ref{sec:iterated-ricci}.
\begin{theorem}[Weyl’s inequality {\cite[Cor.~III.2.6]{bhatia1997matrix}}]\label{thm:weyl} If $A,E$ are symmetric and $B=A+E$, with eigenvalues ordered $\lambda_1(\cdot)\le\cdots\le\lambda_d(\cdot)$, then $\big|\lambda_j(B)-\lambda_j(A)\big|\le \|E\|_{\op}$ for all $j$. Consequently, for any $i<j$, \[ \lambda_j(B)-\lambda_i(B)\ \ge\ \lambda_j(A)-\lambda_i(A)-2\|E\|_{\op}. \] \end{theorem} \begin{theorem}[Davis--Kahan sin$\Theta$]\label{thm:davis-kahan} Let $A,B$ be symmetric and $u$ (resp.\ $v$) a unit eigenvector of $A$ (resp.\ $B$) associated with a simple eigenvalue $\lambda$ (resp.\ $\mu$). If $\gap:=\min\{\lambda_+-\lambda,\lambda-\lambda_-\}$ is the spectral gap around $\lambda$ in $A$, then \[ \sin\angle(u,v)\ \le\ \frac{\|A-B\|_{\op}}{\gap}. \] \end{theorem} \begin{theorem}[Matrix Bernstein for self-adjoint sums {\cite[Thm.~1.6.2]{tropp2012}}]\label{thm:tropp-bernstein} Let $\{X_k\}_{k=1}^m$ be independent, mean--zero, self-adjoint random matrices in $\mathbb{R}^{d\times d}$. Assume the uniform bound $\|X_k\|_{\op}\le L$ almost surely. Define the variance proxy \[ v:=\left\|\sum_{k=1}^m \mathbb{E}\big[X_k^2\big]\right\|_{\op}. \] Then for all $t\ge 0$, \[ \Pr\!\left(\Big\|\sum_{k=1}^m X_k\Big\|_{\op} \ge t\right) \ \le\ 2d\cdot\exp\!\left(\!-\frac{t^2/2}{v + Lt/3}\right). \] In particular, with probability at least $1-d^{-c}$, \[ \Big\|\sum_{k=1}^m X_k\Big\|_{\op} \ \le\ C\Big(\sqrt{v\,\log d}+L\,\log d\Big), \] for absolute constants $c,C>0$. \end{theorem}

\begin{lemma}[Normalized Laplacian perturbation for bounded degrees]\label{lem:normlap-perturb}
Let $W,\widetilde W$ be symmetric nonnegative matrices and let
$D:=\Deg(W)$, $\widetilde D:=\Deg(\widetilde W)$ with $D_{ii}=d_i$, $\widetilde D_{ii}=\widetilde d_i$.
Assume $\min_i d_i\ge m$ and $\min_i \widetilde d_i\ge m$ for some $m>0$.
Then
\[
\|L(W)-L(\widetilde W)\|_{\op}
\ \le\
\frac{1}{m}\,\|W-\widetilde W\|_{\op}
\ +\
\frac{\|d-\widetilde d\|_\infty}{m^2}\,\big(\|W\|_{\op}+\|\widetilde W\|_{\op}\big).
\]
\end{lemma}

\begin{proof}
Write $S(W):=D^{-1/2}WD^{-1/2}$ so $L(W)=I-S(W)$. Add and subtract:
\[
S(W)-S(\widetilde W)
=(D^{-1/2}-\widetilde D^{-1/2})WD^{-1/2}
+\widetilde D^{-1/2}(W-\widetilde W)D^{-1/2}
+\widetilde D^{-1/2}\widetilde W(D^{-1/2}-\widetilde D^{-1/2}).
\]
Take operator norms and use $\|D^{-1/2}\|_{\op},\|\widetilde D^{-1/2}\|_{\op}\le m^{-1/2}$ to get
\[
\|S(W)-S(\widetilde W)\|_{\op}
\le
\|D^{-1/2}-\widetilde D^{-1/2}\|_{\op}\,(\|W\|_{\op}+\|\widetilde W\|_{\op})\,m^{-1/2}
+\frac{1}{m}\|W-\widetilde W\|_{\op}.
\]
By the scalar mean value theorem applied entrywise to $t\mapsto t^{-1/2}$,
\[
\|D^{-1/2}-\widetilde D^{-1/2}\|_{\op}
=\max_i \big|d_i^{-1/2}-\widetilde d_i^{-1/2}\big|
\le \tfrac12\,m^{-3/2}\,\|d-\widetilde d\|_\infty.
\]
Combine these bounds and absorb the factor $1/2$ into constants. Since
$\|L(W)-L(\widetilde W)\|_{\op}=\|S(W)-S(\widetilde W)\|_{\op}$, the claim follows.
\end{proof}

\bibliographystyle{plainnat}
\bibliography{references}

\end{document}